\documentclass[sigconf]{acmart}
\usepackage{multirow}
\usepackage{subfigure}
\usepackage[linesnumbered,ruled, vlined]{algorithm2e}
\usepackage{amsmath, amsthm,amsfonts, bm}
\usepackage{array}
\usepackage{balance}
\usepackage{enumitem}
\usepackage{soul}
\usepackage{marginnote}
\usepackage{tablefootnote}
\usepackage{threeparttable}
\usepackage{footnote}
\usepackage{xcolor}
\makesavenoteenv{tabular}
\makesavenoteenv{table}
\usepackage{url}

\newcolumntype{L}[1]{>{\raggedright\let\newline\\\arraybackslash\hspace{0pt}}m{#1}}
\newcolumntype{C}[1]{>{\centering\let\newline\\\arraybackslash\hspace{0pt}}m{#1}}
\newcolumntype{R}[1]{>{\raggedleft\let\newline\\\arraybackslash\hspace{0pt}}m{#1}}

\newtheorem{definition}{Definition}[section]
\newtheorem{problem}{Problem}[section]

\AtBeginDocument{%
  \providecommand\BibTeX{{%
    \normalfont B\kern-0.5em{\scshape i\kern-0.25em b}\kern-0.8em\TeX}}}

\newcommand{\kaixin}{}
\newcommand{\cheng}{}

\newcommand{\yui}{}

\newcommand{\yuiR}{}

\newcommand{\Revise}{\color{black}}

\begin{document}
\sloppy

\settopmatter{printacmref=false} % Removes citation information below abstract
\renewcommand\footnotetextcopyrightpermission[1]{} % removes footnote with conference information in first column

\title{Maximal Biclique Enumeration with Improved Worst-Case Time Complexity Guarantee: A Partition-Oriented Strategy}

% \pagenumbering{roman}
% \twocolumn{
% \input{cover-letter/thank}
% \input{cover-letter/meta}
% \input{cover-letter/review1}
% \input{cover-letter/review2}
% \input{cover-letter/review3}
% \clearpage
% \nobalance
% }

\author{Kaixin Wang}
\affiliation{%
  \institution{College of Computer Science, Beijing University of Technology, China}
  \country{}
}
\email{kaixin.wang@bjut.edu.cn}

\author{Kaiqiang Yu}
\authornote{Kaiqiang Yu is the corresponding author.}
\affiliation{%
  \institution{State Key Laboratory of Novel Software Technology, Nanjing University, China}
  \country{}
}
\email{kaiqiang.yu@nju.edu.cn}

\author{Cheng Long}
% \authornotemark[1]
\affiliation{%
  \institution{College of Computing and Data Science, Nanyang Technological University, Singapore}
  \country{}
}
\email{c.long@ntu.edu.sg}

\begin{abstract}
    The maximal biclique enumeration problem in bipartite graphs is fundamental and has numerous applications in E-commerce and transaction networks, particularly in areas such as fraud detection and anomaly behavior identification. Most existing studies adopt a branch-and-bound framework, which recursively expands a partial biclique with a vertex until no further vertices can be added. Equipped with a basic pivot selection strategy, all state-of-the-art methods have a worst-case time complexity no better than {\Revise $O(m\cdot (\sqrt{2})^n)$}, where $m$ and $n$ are the number of edges and vertices in the graph, respectively. In this paper, we introduce a new branch-and-bound (BB) algorithm \texttt{IPS}. In \texttt{IPS}, we relax the strict stopping criterion of existing methods by allowing termination when all maximal bicliques within the current branch can be outputted in the time proportional to the number of maximal bicliques inside, reducing the total number of branches required. Second, to fully unleash the power of the new termination condition, we propose an improved pivot selection strategy, which well aligns with the new termination condition to achieve better theoretical and practical performance. 
    Formally, \texttt{IPS} improves the worst-case time complexity to $O(m\cdot \alpha ^n + n\cdot \beta)$, where $\alpha (\approx 1.3954)$ is the largest positive root of $x^4-2x-1=0$ and $\beta$ represents the number of maximal bicliques in the graph, respectively. This result surpasses that of all existing algorithms given that $\alpha$ is strictly smaller than $\sqrt{2}$ and $\beta$ is at most {\Revise $(\sqrt{2})^n-2$} theoretically.
    Furthermore, we apply an inclusion-exclusion-based framework to boost the performance of \texttt{IPS}, improving the worst-case time complexity to $O(n\cdot \gamma^2\cdot\alpha^\gamma + \gamma\cdot \beta)$ for large sparse graphs ($\gamma$ is a parameter satisfying $\gamma \ll n$ for sparse graphs). 
    Finally, we conduct extensive experiments on 15 real datasets and the results demonstrate that our algorithms can run several times faster than the state-of-the-art algorithms on real datasets.
\end{abstract}

% The maximal biclique enumeration problem in bipartite graphs is fundamental and has numerous applications in E-commerce and transaction networks, particularly in areas such as fraud detection and anomaly behavior identification. Most existing studies adopt a branch-and-bound framework, which recursively expands a partial biclique with a vertex until specific termination criteria are met. Equipped with a basic pivot selection strategy, all state-of-the-art methods have a worst-case time complexity no better than $O(m\cdot 2^{n/2})$ (or $O(m\cdot 1.4142^n)$), where $m$ and $n$ are the number of edges and vertices in the graph, respectively. We propose several techniques to boost the efficiency of the algorithm. Finally, we conduct extensive experiments on both real and synthetic datasets, demonstrating that our algorithms can achieve up to an order of magnitude faster than the state-of-the-art on real datasets.




\maketitle
\pagenumbering{arabic}

\section{Introduction}
\label{sec:intro}

{
}

A bipartite graph is a special type of graph where all vertices can be divided into two disjoint sets such that every edge connects a vertex in one set with another vertex in the other set. 
Bipartite graphs are widely used to model interactions between two different classes of objects in many real-world systems, such as purchasing networks between buyers and sellers \cite{wang2006unifying}, collaboration networks between authors and publications \cite{ley2002dblp, zhou2007co}, and recommendation networks between users and items \cite{li2013recommendation}. Given a bipartite graph $G=(L\cup R, E)$, a biclique is a complete subgraph in which every vertex on one side is connected to all vertices on the other side. Among all bicliques, maximal bicliques (MBCs) are of particular interest 
% because they capture the largest dense interaction patterns 
and play a central role in applications including community analysis \cite{lehmann2008biclique, liu2006efficient}, recommendation systems \cite{maier2021biclique, maier2022bipartite}, fraud detection \cite{yu2022efficient}, and graph neural network inference acceleration \cite{jia2020redundancy}.

{\Revise 
\smallskip\noindent\textbf{Motivations and Applications.}
Although the number of maximal bicliques can be exponentially large in the worst case, complete enumeration remains essential for many real-world data analysis tasks, as some examples illustrate below. 

\smallskip\noindent\underline{\textit{Community Discovery and Fraud Pattern Mining.}} 
In many applications, maximal bicliques serve as fundamental structural patterns for discovering dense and meaningful interactions. 
A common modeling paradigm in such applications is to construct a bipartite graph between two types of entities (e.g., users and attributes in social systems, or entities and activities in transaction networks), where a maximal biclique corresponds to a largest group of objects on one side that are tightly associated with the same group of objects on the other side. 
These structures naturally capture coherent behavioral patterns, such as communities of users sharing common attributes in social computing, or groups of entities involved in the same set of suspicious activities in fraud and anomaly detection. 
In both cases, the analytical objective is not merely to identify a single dominant pattern but to uncover all maximal such structures, since each may represent a distinct, meaningful interaction pattern. 
% Partial enumeration may therefore overlook critical structures and lead to incorrect or incomplete analysis, making complete maximal biclique enumeration indispensable in practice.

\smallskip\noindent\underline{\textit{Redundant Message Passing in Graph Neural Networks.}} 
Maximal biclique mining also arises in graph neural networks during message passing between consecutive layers. 
Here, a bipartite graph is formed between the source nodes in one layer and the target nodes in the next, where edges represent the dependencies of the message. 
A maximal biclique corresponds to a largest set of target nodes that receive messages from the same set of source nodes, exposing highly redundant computation patterns. 
Identifying such structures enables factorizing message propagation, reducing memory access, and accelerating training and inference while preserving model accuracy. 
Complete maximal biclique enumeration ensures that all compressible interaction patterns are discovered, providing systematic opportunities for optimizing large-scale GNN systems.

}

\smallskip\noindent\textbf{Existing Methods and Challenges.} 
Quite a few algorithms have been proposed for maximal biclique enumeration \cite{dai2023hereditary, chen2023maximal, abidi2020pivot, chen2022efficient}. These approaches all adopt a similar \emph{branch-and-bound} (BB) framework, which involves recursively dividing the problem of enumerating all MBCs into several  {\yuiR sub-problems} through branching operations. Specifically, each problem corresponds to a search space (or a branch), represented by a triplet $B=(S,C,X)$, where $S$ induces a partial biclique found so far, $C$ contains the vertices that can be added to $S$ to form a larger biclique and $X$ contains the vertices that have already been considered and is used for checking duplications. All MBCs in an input graph $G=(L\cup R, E)$ can be enumerated via the above BB framework by solving the initial search space $B=(\emptyset, L\cup R,\emptyset)$.
Given a branch $B=(S,C,X)$, a set of sub-branches rooted at $B$ would be created, where each sub-branch expands $S$ by adding a vertex $v$ from one side of $C$ to $S$ and updates $C$ and $X$ by removing from $C$ and $X$ those vertices on the other side that are not connected with $v$, respectively. The recursive procedure would terminate when {\kaixin the branch can be \emph{trivially solvable}, i.e., $C$ is empty}, at which point $S$ is reported as a MBC if $X$ is also empty; otherwise, $S$ is not a MBC. 

We term the {\cheng above BB framework, which conducts branching until it reaches \underline{t}rivially \underline{s}olvable branches,}
% branching-until-\underline{t}rivially-\underline{s}olvable-based BB framework 
% for maximal biclique enumeration 
as \texttt{TS}. {\yuiR We note} that the theoretical time complexity and empirical runtime of \texttt{TS}-based algorithms are dominated by the branching process, as the number of branches can become exponential in the worst case~\cite{dai2023hereditary, chen2023maximal, abidi2020pivot, chen2022efficient}. Therefore, existing studies of \texttt{TS} target two directions to improve performance by reducing the number of branches created. 
One direction is based on the observation that all MBCs in a bipartite graph $G=(L\cup R, E)$ can be enumerated only by visiting some subsets {\cheng in} $L$'s powerset \cite{chen2023maximal, zhang2014finding}. The rationale is that given a vertex set $S_L\in P_L$ where $P_L$ is the powerset of $L$ and let $S_R$ contain all common neighbors of the vertices in $S_L$. Then $S_L\cup S_R$ induces an MBC if $S_L$ also contains all common neighbors of the vertices in $S_R$. Therefore, it is common practice to choose the vertex set with the minimum number of vertices inside as $L$ and {\yuiR visit its}  powerset recursively \cite{chen2023maximal, zhang2014finding}. 
The other direction is to prune more branches effectively by \emph{pivoting} strategy \cite{abidi2020pivot, dai2023hereditary, chen2022efficient}. The main idea is to first choose a vertex from $C\cup X$ as the pivot, and then prune those branches that would not {\yuiR contain} any MBC inside based on some topological information of the graph, such as the containment relationship \cite{abidi2020pivot, chen2022efficient} or the neighborhood information \cite{dai2023hereditary}. For example, given a branch $B=(S,C,X)$, the state-of-the-art pivoting strategy first chooses a vertex $v\in C\cup X$ as the pivot and the sub-branches produced by including (1) the vertex at the same side as that of $v$ and (2) $v$'s neighbors at the other side can be safely pruned \cite{dai2023hereditary}. To prune the maximum number of branches, the common practice is to choose the vertex with the maximum degree from $C\cup X$ as the pivot. By {\cheng adopting} the above techniques, existing algorithms achieve the state-of-the-art worst-case time complexity of $O(m\cdot (\sqrt{2})^n)$, where $m$ and $n$ denotes the number of edges and vertices of the graph, respectively. 
{\cheng We note that it is challenging to further improve beyond this worst-case time complexity with this framework since a graph may contain up to $(\sqrt{2})^n-2$ MBCs, which implies that the total number of trivially solvable branches within \texttt{TS} could be $O((\sqrt{2})^n)$ in the worst case.
}

% Although 
% {\cheng While} the above optimization techniques have been proved to be effective in the efficiency improvement of \texttt{TS}, it remains challenging to further improve the current worst-case time complexity of $O^*((\sqrt{2})^n)$. 
% This is because a graph may contain up to $(\sqrt{2})^n-2$ MBCs, which implies that the total number of trivially solvable branches within \texttt{TS} is inherently bounded by $O((\sqrt{2})^n)$ in the worst case.}

\smallskip\noindent\textbf{Our Methods.}
% To solve the aforementioned issue, 
% in this paper, we propose to 
% {\cheng In this paper, we propose a new framework which improves }
% improve the framework by reducing the time costs for the branching procedure in both theory and practice. 
We observe that the termination condition of \texttt{TS}, which requires the branch to be trivially solvable, may be overly strict and could lead to significant computational overhead during branching. In fact, it would be sufficient if the branch can be solvable in optimal time - referred to as being \emph{optimally solvable}. Specifically, a problem instance with input graph $G$ is considered optimally solvable if it can be solved in $O(scan(G) + scan(\mathcal{B}_G))$ time, where $O(scan(G))$ denotes the time required to load the graph $G$, $\mathcal{B}_G$ represents the set of all MBCs in $G$, and $O(scan(\mathcal{B}_G))$ corresponds to the time needed to output all MBCs in $G$. 
Motivated by the above observation, 
{\cheng we propose a new BB framework termed \texttt{OS}, 
which conducts branching until it reaches \underline{o}ptimally \underline{s}olvable branches}.
% termed the branching-until-\underline{o}ptimally-\underline{s}olvable-based framework \texttt{OS}, would recursively divide each branch into multiple sub-branches until they can be solved optimally. 
Given the fact that an optimally solvable branch can contain multiple MBCs, while a trivially solvable branch can only contain at most one MBC, 
% the total number of optimally solvable branches is typically smaller than that of trivially solvable ones, 
{\kaixin it allows \texttt{OS} to (1) produce much less total number of branches during enumeration, making it possible to improve the worst-case time complexity theoretically and (2) reduce the computational overhead across the branching process, making it more efficient practically compared to the \texttt{TS} framework.}
To develop a practical algorithm based on the \texttt{OS} framework, two key questions must be addressed: (1) what types of graphs satisfy the optimally-solvable condition? (2) how should the BB process be designed to ensure it fully unleashes the power of the optimally solvable condition?

To address the first question, we observe that there exists a kind of bipartite subgraphs known as 2-biplexes~\footnote{A bipartite graph $H$ is a $k$-biplex if each vertex on one side has at most $k$ non-neighbors on the other side.}, which allows direct construction and enumeration of all MBCs inside in nearly optimal time, without requiring further branching. The core idea is to convert the identified 2-biplex $H$ into its complementary bipartite graph, $\overline{H}$~\footnote{Given a bipartite graph $H = (L \cup R, E)$, its complementary bipartite graph $\overline{H} = (L \cup R, \overline{E})$ shares the same vertex sets as $H$, with an edge between $u \in L$ and $v \in R$ if and only if $(u, v) \notin E$, i.e., $\overline{E} = (L \times R) \setminus E$.}. Since $\overline{H}$ contains only isolated vertices, simple paths, and simple cycles of even length (since its maximum degree is at most 2), each MBC in $H$ corresponds to a maximal independent set in $\overline{H}$. These independent sets can be enumerated iteratively, ensuring that no additional computational cost is incurred during enumeration. {\kaixin We note that it still remains an open question whether the termination from 2-biplexes can be further relaxed to other bipartite subgraphs (e.g., $k$-biplexes with $k\ge 3$) that can be optimally solved.}

To address the second question, we aim to integrate the new termination condition into the BB framework. An intuitive idea is to directly build upon the state-of-the-art BB algorithms. Unlike the existing studies that only terminate the recursive process when both vertex sets $C$ and $X$ are empty, the new design, which follows \texttt{OS} framework, would terminate the process as soon as the vertices in $S\cup C$ form a 2-biplex and $X$ becomes empty. This ensures that every MBC identified within the 2-biplex induced by $S\cup C$ is also a MBC in the whole input graph $G$.
%
% However, despite the efficiency gained through this new termination condition, 
{\cheng Nevertheless, the above design may still generate} much more branches than necessary. The reason is that the pivoting strategy in \cite{dai2023hereditary} is not well-aligned with the early-termination technique, especially when the vertices that are chosen as the pivots have only 1 or 2 non-neighbor(s) on the other side. These vertices may {\cheng be} possibly contained in a 2-biplex, and if a different vertex were chosen as the pivot, the termination condition could have been applied earlier to the produced branches.
Motivated by this, we propose an improved pivot strategy based on a new partition-based framework to solve the above issue. The idea is to \emph{skip} some vertices in $C$ and choose the pivot vertex from the remaining vertices. We determine the partition in each branch $B$ such that when no pivot can be selected in $B$, $B$ is optimally-solvable and the new termination condition can be naturally applied. We present an algorithm called \texttt{IPS} with this new pivoting strategy. Formally, we prove that \texttt{IPS} would produce at most $O(\alpha^n)$ branches, where $\alpha\approx 1.3954$ is the largest positive root of $x^4-2x-1=0$ and enumerate all MBCs in $O(m\cdot \alpha^n+n\cdot \beta)$ time under the worst case, where $\beta$ represents the number of MBCs in the graph, i.e., $\beta = |\mathcal{B}_G|$. It is strictly better than that of the existing algorithms since both $O(m\cdot \alpha^n)$ and $O(n\cdot \beta)$ are bounded by $O(m\cdot (\sqrt{2})^n)$ given the fact that $\beta$ is at most $(\sqrt{2})^n-2$ for any graph with $n$ vertices.

In addition, we adapt an existing technique, called \emph{inclusion-exclusion} (\texttt{IE}), for boosting the efficiency and scalability of the BB algorithms including \texttt{IPS}. The idea is to construct multiple problem instances each on a smaller subgraph, which has been widely used for enumerating and finding subgraph structures \cite{chen2021efficient, wang2022listing, yu2023maximum}.

\smallskip
\noindent\textbf{Contributions.} We summarize our contributions as follows.

\begin{itemize}[leftmargin=*]
    \item We propose a new branch-and-bound algorithm called \texttt{IPS} for maximal biclique enumeration problem, which is based on our newly developed termination condition and the new pivoting strategy. \texttt{IPS} has a strictly better worst-case time complexity than that of the state-of-the-arts, i.e., the former is $O(m\cdot \alpha^n+n\cdot \beta)$ and the latter is $O(m\cdot (\sqrt{2})^n)$. (Section~\ref{sec:BB-v2})
    % , where $\beta$ is the number of maximal bicliques in the graph.  
    
    \item We further introduce an inclusion-exclusion framework (\texttt{IE}) to boost the performance of \texttt{IPS}. When \texttt{IE} is applied to \texttt{IPS}, the worst-case time complexity becomes $O(n\cdot\gamma^2\cdot \alpha^{\gamma}+\gamma\cdot \beta)$, where $\gamma$ is a parameter satisfying $\gamma\ll n$ for sparse graphs. This is better than that of \texttt{IPS} when the graph satisfies the condition $n\ge \gamma+6.01\ln \gamma$. 
    % We note that the condition is satisfied by all existing real-world graphs due to their sparse nature. 
    (Section~\ref{sec:boost})
    \item We conduct extensive experiments on both real and synthetic datasets to evaluate the effectiveness of our proposed techniques, and the results show that our algorithm can consistently and largely outperform the state-of-the-art algorithms. (Section~\ref{sec:exp})
\end{itemize}

The rest of the paper is organized as follows. Section~\ref{sec:problem} reviews the problem and presents some preliminaries. Section~\ref{sec:related} reviews the related work and Section~\ref{sec:conclusion} concludes the paper.
% \clearpage
\section{Problem and Preliminaries}
\label{sec:problem}

Let $G=(L\cup R, E)$ be an unweighted and undirected bipartite graph, where $L$ and $R$ are two disjoint vertex sets and $E$ is the edge set. Given $V_L\subseteq L$ and $V_R \subseteq R$, we use $H=G[V_L \cup V_R]$ to denote the subgraph of $G$ induced by the vertices in $V_L\cup V_R$. For a subgraph $H$ of a bipartite graph $G$, we use $L(H)$, $R(H)$ and $E(H)$ to denote the left side, right side and set of edges {\cheng of $H$}, respectively. We note that all subgraphs considered in this paper are induced subgraphs. 
% We use $H$ or $X\cup Y$ as a shorthand of the subgraph $H = G[X\cup Y]$, and use $(|L(H)|,|R(H)|)$ to denote the size of $H$.
Given a vertex $u\in L$ and a vertex set $V_R\subseteq R$, we use $N(u, V_R)$ (resp. $\overline{N}(u, V_R)$) to denote the set of neighbors (resp. non-neighbors) of $u$ in $V_R$. We define $d(u, V_R) = |N(u, V_R)|$ and $\overline{d}(u, V_R) = |\overline{N}(u, V_R)|$. 
% $V'\subseteq (L\cup R)$ and a vertex $u$ from one side of $G$, we use $N(u, V')$ (resp. $\overline{N}(u, V')$) to denote the set of neighbours (resp. non-neighbours) of $u$ from the other side. We define $d(u, V') = |N(u, V')|$ and $\overline{d}(u, V') = |\overline{N}(u, V')|$. 
Let $V_L\subseteq L$ be a set of vertices, we use $N(V_L, V_R)$ to denote the set of common neighbors of $V_L$ in $V_R$, i.e., $N(V_L, V_R)=\cap_{u\in V_L} N(u,V_R)$. 
We have symmetric definitions for the vertices in $R$.

\begin{definition}[Biclique \cite{bondy1976graph}]
    Given a graph $G=(L\cup R, E)$, a subgraph $H = G[V_L\cup V_R]$ of $G$ is said to be a biclique if $\forall u\in V_L$ and $\forall v\in V_R$, there is an edge $(u, v)\in E$.
\end{definition}

Given a bipartite graph $G$, a biclique $H$ is said to be maximal if there is no other biclique $H'$ that contains $H$. Among all maximal bicliques of $G$, those maximal bicliques with the more edges are of more interests since they carry more useful information in fraud detection or community search tasks \cite{sim2009mining}. Therefore, in this paper, we target the maximal biclique enumeration problem but also consider some additional constraints that can be specified by the {\cheng users}. In specific, we consider two size constraints $\tau_L$ and $\tau_R$ such that for each returned maximal biclique $H$, it satisfies $|L(H)|\ge \tau_L$ and $|R(H)|\ge \tau_R$. These constraints help filter out those skewed bicliques, i.e., the number of vertices on one side is significantly larger than that on the other side, and those small-size bicliques. {\cheng We} formally present the problem definition as follows.

\begin{problem}[Maximal Biclique Enumeration \cite{dai2023hereditary, chen2023maximal}]
    Given a bipartite graph $G=(L\cup R, E)$ and two integers $\tau_L\ge 1$ and $\tau_R \ge 1$, the maximal biclique enumeration problem aims to return a set of maximal bicliques such that each maximal biclique $H$ inside satisfies $|L(H)|\ge \tau_L$ and $|R(H)|\ge \tau_R$. 
\end{problem}

% \smallskip\noindent
% \textbf{NP-hardness}. The Top-$k$ Maximum Biclique Search problem is NP-hard since it can reduce to another NP-Hard problem, namely Maximum Biclique Search (Max-BCS in short) problem \cite{lyu2020maximum}, in polynomial time, where the Max-BCS problem aims to find the maximal biclique with the maximum number of edges inside. The proof of the NP-hardness of the Max-BCS problem can be found in \cite{peeters2003maximum}. The reduction can be done in $O(1)$ by setting $k=\tau_L=\tau_R=1$. 

\smallskip\noindent
\textbf{NP-hardness.} The problem is NP-hard since the problem reduces to an existing NP-hard problem, namely the maximum biclique search problem (i.e., the problem aims to find a biclique $H^*$ with the most number of edges inside) \cite{lyu2022maximum}, if $\tau_L$ and $\tau_R$ are initialized to be $|L(H^*)|$ and $|R(H^*)|$, respectively.

% As discussed above, by setting $k=1$, the Topk-BCS problem would be equivalent to the Maximum Biclique Search problem. Besides, by setting $k$ a extremely large number,
% % e.g., $k\ge 2^{n/2}$ ($n$ is the number of vertices in $G$), 
% the Topk-BCS problem would return all maximal bicliques in $G$, which would be equivalent to the Maximal Biclique Enumeration (MBCE in short) problem \cite{abidi2020pivot, chen2022efficient, dai2023hereditary}. 
% %
% Therefore, in the Section~\ref{sec:MBCS}, we mainly focus on the former case, i.e., $k=1$. The algorithm can be naturally extended to arbitrary $k$ with some modifications, which we will present the dicussion in Section~\ref{sec:MBCE}.

\section{The Branch-and-Bound Algorithms}
\label{sec:BB-v2}

\subsection{Overview of the Branch-and-bound Algorithms}
\label{subsec:overview}

Many algorithms have been proposed for maximal biclique enumeration problem in the literature \cite{dai2023hereditary, abidi2020pivot, chen2022efficient, chen2023maximal}. Most of them \cite{abidi2020pivot, chen2023maximal, dai2023hereditary} adopt a conventional and widely-used branching strategy that we call Bron-Kerbosch (BK) branching \cite{bron1973algorithm}.
Specifically, it recursively partitions the search space (i.e., the set of all maximal bicliques in $G$) into several sub-spaces via branching until some stopping criterion are satisfied, e.g., $C$ and $X$ become empty sets \cite{dai2023hereditary,chen2021efficient}. Each search space, which corresponds to a branch $B$, can be represented as a triplet $B=(S,C,X)$, where 
\begin{itemize}[leftmargin=*]
    \item \textbf{Set $S$} induces a partial biclique found so far;
    \item \textbf{Set $C$} contains all candidate vertices that can be included to $S$  within the search space {\cheng for expanding the biclique}; {\cheng and}
    \item \textbf{Set $X$} contains all exclusion vertices that must be excluded from $S$  within the search space.
\end{itemize}
We further denote the left side and right side of $S$ (resp. $C$ and $X$) by $S_L$ (resp. $C_L$ and $X_L$) and $S_R$ (resp. $C_R$ and $X_R$), respectively.

\begin{algorithm}[t]
\small
\caption{The general framework: \texttt{BasicBB}}
\label{alg:basic-pivot}
\KwIn{A bipartite graph $G=(L\cup R,E)$, $\tau_L$ and $\tau_R$.}
\KwOut{All maximal bicliques in $G$ that satisfy the size constraints.}
\SetKwFunction{List}{\textsf{BB\_Rec}}
\SetKwProg{Fn}{Procedure}{}{}
% Let $v_1, \cdots, v_n$ be in the degeneracy ordering of $G$\;
% \For{each $v_i\in V(G)$}{
% $P_i\gets N(v_i)\cap \{v_1, \cdots, v_{i-1}\}$\;
% $X_i\gets N(v_i)\cap \{v_{i+1}, \cdots, v_n\}$\;
% \List{$\{v_i\}, P_i, X_i$}\;
% }
% $H^*\gets G[\emptyset]$ \tcp*[l]{Global variable}
\List{$\emptyset, L\cup R, \emptyset$}\;
\Fn{\List{$S, C, X$}} {
    \tcc{Termination}
    \If{$C$ and $X$ satisfy any stopping criterion}{
        \textbf{Output} all maximal bicliques in $G[S\cup C]$; 
        \Return\;
    }
    \tcc{Pruning}
    \If{any of pruning rule can be applied}{\Return\;}
    \tcc{Branching with Pivoting}
    Choose a pivot vertex $v^*$ from $C\cup X$\;
    Create a set of branches $B_i$ based on the pivot vertex $v^*$\;
    \For{each created branch $B_i$}{
    
    \List{$S_i, C_i, X_i$}\;
    }
    % \tcc{Branching with basic pivot strategy}
    % $v^*\gets$ a vertex in $C\cup X$ with the min. value of $\overline{d}(\cdot, C)$\;
    % \uIf{$v^*\in X$}{Create $\overline{d}(v^*, C)$ branches $B_1, B_2, \cdots$ over the vertices in $\overline{N}(v^*, C)$ based on Eq.~(\ref{eq:branching})\;}
    % \Else{Create $\overline{d}(v^*, C) + 1$ branches $B_1, B_2, \cdots$ over the vertices in $\{v^*\}\cup \overline{N}(v^*, C)$ based on Eq.~(\ref{eq:branching})\;}
    % \For{each created branch $B_i$}{\List{$S_i, C_i, X_i$}\;}
}
\end{algorithm}

\begin{figure}[t]
% \vspace{-2mm}
\centering
% \subfigure[A 3-biplex graph $G$]{
    \includegraphics[width=0.38\textwidth]{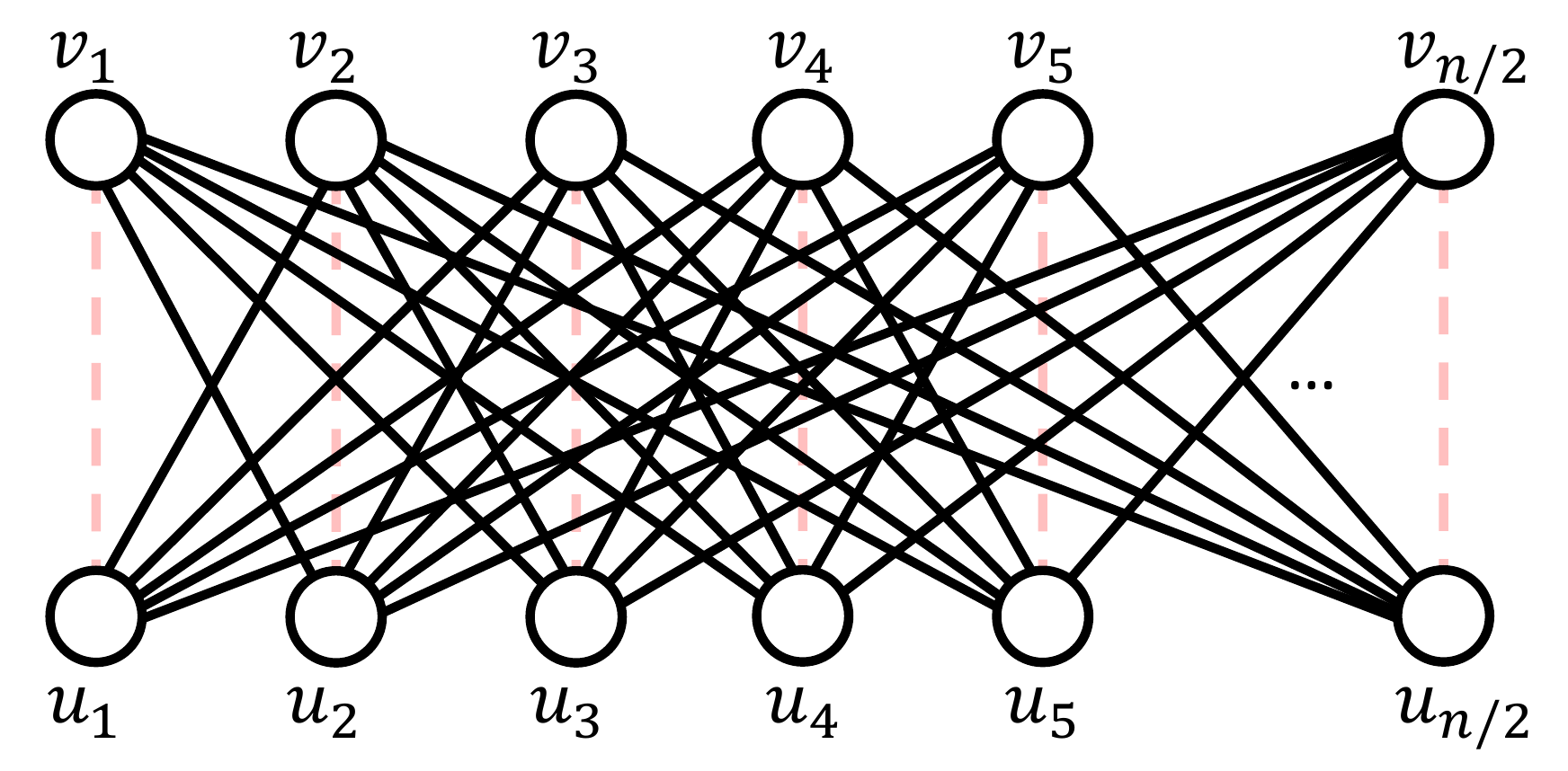}
% \vspace{-3mm}
    % }
% \subfigure[The complementary of the candidate graph after branching on $u_1$ with another pivot strategy, e.g., $u_1$ is the pivot vertex.]{
% \label{subfig:branch2}
%     \includegraphics[width=0.255\textwidth]{example/branch2.pdf}}
\caption{\Revise{An example graph with $n$ vertices ($n$ is even), in which each $u_i$ is disconnected only from $v_i$ for $i=1,2,\ldots,n/2$, as indicated by the dashed edges. In total, the graph contains exactly $(\sqrt{2})^n - 2$ maximal bicliques. To see this, observe that any non-empty subset of vertices chosen from $U$ uniquely determines a maximal biclique by taking its common neighbors in $V$. Since $|U|=|V|=n/2$, the number of such choices is $2^{n/2}-2$ (equiv.\ $(\sqrt{2})^n - 2$), excluding the empty set and the full set so that both sides of the biclique are non-empty.
}}
% \vspace{-2mm}
\label{fig:existing}
\end{figure}

The recursive process starts from the initial branch $B=(S,C,X)$ with $S=\emptyset$, $C=L\cup R$ and $X=\emptyset$. Consider the branching step at a branch $B=(S,C,X)$. Let $\langle v_1, v_2, \cdots, v_{|C|} \rangle$ be an ordering of the vertices in $C$. Then the branching step would create $|C|$ sub-branches, where the $i$-th branch, denoted by $B_i$, includes $v_i$ to $S$ and excludes $v_1, v_2, \cdots v_{i-1}$ from $S$. Formally, for $1\le i\le |C|$, the sub-branch $B_i$ can be represented by
\begin{equation}
\label{eq:branching}
    S_i = S\cup \{v_i\}\quad C_i = C\setminus\{v_1, \cdots, v_i\}\quad X_i = X\cup \{v_1, \cdots, v_{i-1}\}
\end{equation}
We note that $C_i$ and $X_i$ can be further refined by removing those vertices that are not $v_i$'s neighbors on its opposite side since they cannot form bicliques with $v_i$. The above branching step corresponds to a binary search process, where one searches those maximal bicliques including $v_i$ by adding $v_i$ from $C$ to $S$ and the other searches those maximal bicliques excluding $v_i$ by removing $v_i$ from $C$ to $X$. 
% It first partitions the branch $B$ into two branches based on $v_1$ - one including $v_1$ to $S$ (this is $B_1$) and the other excluding $v_1$ from $S$. It further partitions the latter into two branches based on $v_2$ - one including $v_2$ (this is $B_2$) and the other excluding $v_2$. It continues the process until the last branch $B_{|C|}$, which includes $v_{|C|}$ and excluding $v_1, v_2, \cdots, v_{|C|-1}$, is created. 
Consequently, the total number of branches can grow to $O(2^n)$ in the worst case for a graph $G$ with $n$ vertices. To reduce the number of explored branches during the search process, existing branch-and-bound algorithms for maximal biclique enumeration problem adopt different design strategies, which mainly differ in the following aspects.
% As a result, there would exist $O(2^n)$ branches under the worst case in a graph $G$ with $n$ vertices. 
%
% To reduce the number of branches, various branch-and-bound algorithms for maximal biclique enumeration differ from each other in the following aspects. 
\begin{itemize}[leftmargin=*]
    \item \textbf{Stopping Criterion} terminates the recursive procedure and outputs all maximal bicliques within the branch (if any), i.e., those within the subgraph induced by $G[S\cup C]$;
    
    \item \textbf{Pivoting Strategy} {\Revise selects a vertex from one side of $C\cup X$ as the pivot and produces sub-branches only on the pivot itself and its non-neighbors on the other side, which ensures that all maximal bicliques are enumerated exactly once;} 
    % (i.e., mutually exclusive collectively exhaustive); }
    % determines which vertex should be selected from $C\cup X$ as the pivot, and the selection further determines how the branches are produced;    

    \item \textbf{Pruning Technique} eliminates the unpromising branch, i.e., the branch not containing any maximal bicliques,  and all sub-branches rooted at $B$. 
\end{itemize}

% Specifically, state-of-the-art algorithms \cite{dai2023hereditary} terminate the recursive procedure when $C$ and $X$ are empty and output $G[S]$ as a maximal biclique. In addition, given a branch $B=(S,C,X)$ during recursion, the pivot $v^*$ is selected from $C\cup X$ by choosing one with the most number of neighbors in $C$ since this strategy would prune as many branches as possible.
% strategy of \cite{dai2023hereditary} first chooses a vertex $v\in C\cup X$ as the pivot and the sub-branches produced by including (1) vertices on the same side as those of $v^*$ and (2) $v$'s neighbors on the other side of $C$ can be safely pruned \cite{dai2023hereditary}. 
% To prune the maximum number of branches, the common practice is to choose the vertex with the maximum degree from $C\cup X$ as the pivot. 
% By adopting the above techniques, existing algorithms achieve the worst-case time complexity of $O(m\cdot \sqrt{2}^n)$, where $m$ and $n$ denote the number of edges and vertices of the graph, respectively. The worst case occurs when the pivot selected in every branch $B$ during recursion has only one non-neighbor on the other side. An example is provided in Figure~\ref{}. It should be noted that the example graph just contains $\sqrt{2}^n -2$ MBCs, making it hard to improve beyond this worst-case time complexity under the basic framework. 

{\Revise Specifically, state-of-the-art algorithms~\cite{dai2023hereditary} follow the \texttt{BasicBB} framework. 
In particular, they terminate the recursive procedure when both $C$ and $X$ become empty and output $G[S]$ as a maximal biclique. 
Moreover, for a branch $B=(S,C,X)$ during recursion, the pivot $v^*$ is selected from $C_L \cup X_L$ (or $C_R \cup X_R$) as the vertex with the minimum number of non-neighbors in $C_R$ (or $C_L$), so as to produce as few sub-branches as possible.
By adopting these techniques, the algorithms achieve a worst-case time complexity of $O(m \cdot (\sqrt{2})^n)$, where $m$ and $n$ denote the number of edges and vertices of the graph, respectively~\cite{dai2023hereditary}. 
This worst case occurs when the pivot selected in every branch $B$ has only one non-neighbor. 
Figure~\ref{fig:existing} illustrates such an example. 
Notably, the example graph contains exactly $(\sqrt{2})^n - 2$ maximal bicliques, indicating that the existing bound of the time complexity under the basic framework is nearly worst-case optimal and it is difficult to further improve it.

}

\if 0
To reduce the number of branches, existing study applies the \emph{pivoting} strategy \cite{dai2023hereditary}. Let $B=(S,C,X)$ be a branch and $v^*\in C_L\cup X_L$ be a pivot vertex. All sub-branches produced by including the vertex in $C_L\setminus \{v^*\}$ and $N(v^*, C_R)$ can be safely pruned. Note that a set of branches can be symmetrically pruned if the pivot vertex is selected from $C_R \cup X_R$. We denote the this  branch-and-bound framework as \texttt{BasicBB} and summarize the pseudo-code in Algorithm~\ref{alg:basic-pivot}. 
% The illustration of \texttt{ParBB} and its relationship to \texttt{BasicBB} is shown in Figure~\ref{fig:illustration}.
% based on the following lemma. 
\fi

In this paper, we propose new techniques for maximal biclique enumeration, including a new stopping criterion and a new pivoting strategy. Armed with a new partition-based branch-and-bound framework, these two techniques collectively contribute to significant theoretical advancements. Specifically, we first elaborate on how the partition-based branch-and-bound framework operates  and how our proposed techniques would be seamlessly integrated into the framework (Section~\ref{subsec:motivation}). We further summarize all possible pruning techniques to the problem setting of enumerating size-constrained maximal bicliques. Finally, we present that an algorithm, incorporating these techniques, achieves a worst-case time complexity of $O(m \cdot \alpha^n + n \cdot \beta)$, where $\alpha (\approx 1.3954)$ is the largest positive root of $x^4 - 2x - 1 = 0$ and $\beta$ is the number of maximal bicliques in the graph. This result outperforms existing state-of-the-art branch-and-bound-based algorithms (Section~\ref{subsec:summary}).

% , whose worst-case time complexity is $O(m \cdot\sqrt{2}^n)$ 

\begin{algorithm}[t]
\small
\caption{The partition-based framework: \texttt{ParBB}}
\label{alg:impr-bb}
\KwIn{A bipartite graph $G=(L\cup R,E)$, $\tau_L$ and $\tau_R$.}
\KwOut{All maximal bicliques in $G$ that satisfy the size constraints.}
\SetKwFunction{List}{\textsf{ParBB\_Rec}}
\SetKwProg{Fn}{Procedure}{}{}
\List{$\emptyset, L\cup R, \emptyset$}\;
\Fn{\List{$S, C, X$}} {
    \tcc{Termination}
    \If{$C$ and $X$ satisfy any stopping criterion}{
        \textbf{Output} all maximal bicliques in $G[S\cup C]$; 
        \Return\;
    }
    \tcc{Pruning}
    \If{any of pruning {\cheng rules} can be applied}{\Return\;}

    \tcc{Partition-based Branching with Pivoting}
    {\cheng Partition} $C$ into two sets $C'$ and $C\setminus C'$, where $C'\subsetneq C$\;

    Choose a pivot vertex $v^*$ from $C \cup X$\;
    Create a set of branches $B_i$ based on the pivot vertex $v^*$ and the partition\;
    \For{each created branch $B_i$}{
    
    \List{$S_i, C_i, X_i$}\;
    }
}
\end{algorithm}

% \section{A New Branch-and-bound Framework}
% \label{sec:impr-BB}

\subsection{A Partition-based Framework: \texttt{ParBB}}
\label{subsec:motivation}

% We motivate our technique from the following two perspectives. 
%
% \underline{From the perspective of the problem}, 
% Many existing studies have witnessed a common phenomenon in real-world graphs, including the bipartite graphs that we target in this paper: they are typically globally sparse but locally dense \cite{chang2019cohesive, boccaletti2006complex}. This characteristic is further validated by the presence of numerous maximal bicliques within these graphs, where a significant portion of these bicliques overlap extensively with one another. The substantial overlap among these maximal bicliques not only leads to redundant computations but also contributes to the formation of dense regions within the graph.
%
% Therefore, an important question arises: can we develop more efficient methods to reduce the redundant computations associated with these overlapping maximal bicliques? Additionally, can we leverage the dense regions induced by these overlaps to gain deeper insights and optimize computations? 
% To answer these questions, we first introduce a new branch-and-bound framework called \texttt{ParBB}, which brings some insights. With a carefully-designed, we can prune a batch of non-maximal bicliques ... 

% To address these two questions, 
\subsubsection{Overview}
In this section, we introduce a new branch-and-bound algorithm called \texttt{ParBB}. 
% , which follows the idea of \emph{divide-and-conquer}. 
Different from \texttt{BasicBB}, given a branch $B=(S,C,X)$, \texttt{ParBB} partitions the candidate set $C$ into two disjoint sets $C'$ and $C\setminus C'$ where $C'\subsetneq C$. 
The key advantage of this partitioning lies in the ability to defer the computation of maximal bicliques involving only vertices in $C'$. Specifically, instead of processing all vertices in one step, \texttt{ParBB} allows the enumeration of maximal bicliques in $C'$ at a later time, in a batch-based manner. This flexibility, enabled by partitioning, ensures that the maximal bicliques in $C'$ are handled efficiently without unnecessary recomputation. By contrast, non-partition-based frameworks, like \texttt{BasicBB}, do not provide such deferred handling and must process all candidates immediately.
We also emphasize that this partition-based framework generalizes \texttt{BasicBB}. If $C'$ is set to be empty at every branch, \texttt{ParBB} reduces to \texttt{BasicBB}. However, the partitioning strategy provides significant benefits by allowing deferred, batch-based enumeration and pruning, leading to a more efficient handling of maximal bicliques compared to non-partition-based frameworks, which must process each biclique individually. In the following, we explore two partitioning strategies that reduce computational overhead by enabling batch-based output of maximal bicliques and batch-based pruning of non-maximal ones. We summarize the procedure following the partition-based framework in Algorithm~\ref{alg:impr-bb}. 

\subsubsection{Candidate Set Partition Strategies}
\label{subsubsec:partition-non-empty}
The intuition behind our partition strategies stems from a key non-trivial observation, which we summarize in the following lemma.
\begin{lemma}
    Given a bipartite graph $G=(L\cup R, E)$ being a 2-biplex, all maximal bicliques in $G$ can be enumerated in $O(m+n\cdot \beta)$ time, where $m=|E|$, $n=|L|+|R|$ and $\beta$ is the exact number of maximal bicliques in $G$, respectively.
\end{lemma}
\begin{proof}
    To maintain coherence, we omit the detailed proof here. The idea is that each connected component of the complementary graph $\overline{G}$ of a 2-biplex $G$ is either a single vertex, a simple path, or a simple cycle. This follows naturally, as the maximum degree of $\overline{G}$ is at most 2. Consequently, each maximal biclique in $G$ corresponds to a maximal independent set in $\overline{G}$, which can be constructed in a batch-based manner by merging maximal independent sets from the single vertices, simple paths, and simple cycles. 
    % An algorithm verifying this lemma is provided in the final part of this section.
\end{proof}

% is to find a computationally efficient method to identify a subset of candidate vertices $C'$. To achieve this goal, 
% such that any bicliques in $G[S\cup C']$ cannot be globally maximal. 
% To achieve this goal, we observe that $X$ contains the vertices that have been explored in earlier branches and are typically used to check maximality. Therefore, we leverage the number of connections and/or disconnections between the vertices in $C$ and $X$ to determine our partition. 

This lemma tells us that there exists such a group of graphs that all maximal bicliques inside can be enumerated in nearly optimal time since it needs at least $O(n\cdot \beta)$ to output all maximal bicliques in $G$ and the algorithm only takes extra $O(m)$ time. To equip this lemma into our algorithm, we propose a new stopping criterion shown as follows. 
% \begin{minipage}{0.45\textwidth}
\vspace{2mm}

\noindent
\fbox{
\begin{minipage}{0.45\textwidth}
\textbf{Stopping Criterion.} If a branch $B=(S,C,X)$ satisfies $G[S\cup C]$ induces a 2-biplex and $X=\emptyset$,  we output all maximal bicliques in $G[S\cup C]$ and return.
\end{minipage}
}
% \end{minipage}
\vspace{2mm}
% However, this lemma cannot be directly applied to our problem due to three reasons. First, 

Consider an arbitrary branch $B = (S, C, X)$ that does not satisfy the stopping criterion. It must fall into at least one of the following two cases: (1) there exist vertices in $C_L$ with at least 3 non-neighbors in $C_R$ (or symmetrically, vertices in $C_R$ with at least 3 non-neighbors in $C_L$), or (2) there exist vertices in $X$. 
Therefore, to minimize the number of produced branches, it is essential to apply the stopping criterion as early as possible, which aims to exhaustively eliminate vertices belonging to these cases as quickly as possible.
% To fully harness the power of this termination condition, the key step is to exhaustively eliminate the presence of vertices belonging to these cases as quickly as possible. 
To achieve this, we divide all non-terminable branches into two categories based on whether $X$ is empty, and we propose a candidate partition strategy and enumeration procedure for each category as follows.

\smallskip\noindent\textbf{Category 1 - $\bm{X}$ is empty.} When $X$ is empty, each maximal biclique in the subgraph $G[S\cup C]$ is also a maximal biclique of $G$. 
% We first consider the simple case when $X$ is empty. 
We partition the candidate set into two disjoint sets such that $C'$ always induces a 2-biplex. Formally, for each $v\in C_L$, we collect the number of disconnections with the vertices in $C_R$, i.e., $\overline{d}(v, C_R)$. Symmetrically, for each $v\in C_R$, we collect the number of disconnections with the vertices in $C_L$, i.e., $\overline{d}(v, C_L)$. We define the partition $C'=C_L'\cup C_R'$ as follows, where 
\begin{equation}
\label{eq:C_L_C}
\begin{aligned}
    C_L'=\{v\in C_L\mid \overline{d}(v, C_R) \le 2\}    \\
% \end{aligned}
% % \end{equation}
% % and
% % \begin{equation}
% % \label{eq:C_R_C}
% \begin{aligned}
     C_R'=\{v\in C_R\mid \overline{d}(v, C_L)\le 2 \} 
\end{aligned}
\end{equation}
% The rationale of the above definition is as follows. First, . 
When $C_L'$ and $C_R'$ are not empty, it can be easily verified the graph induced by $C'$ is a 2-biplex since $C_R'$ is a subset of $C_R$, for a vertex $v\in C_L'$, we have $\overline{d}(v, C_R')\le \overline{d}(v, C_R)\le 2$. Similarly, since $C_L'$ is a subset of $C_L$,  for a vertex $v\in C_R'$, we have $\overline{d}(v, C_L')\le \overline{d}(v, C_L)\le 2$.

Recall that (1) our goal is to reduce the presence of vertices with at least 3 non-neighbors in $C_L$ and $C_R$ (since $X$ is empty in this category), i.e., those vertices in $C_L\setminus C_L'$ and $C_R\setminus C_R'$, and (2) the pivoting strategy in \cite{dai2023hereditary} says that the branching steps conducted on the pivot vertex $v^*$ and its non-neighbors are already enough to enumeration all maximal biclique within a branch, We have two choices.  One is to choose a vertex $v^*\in C_L\setminus C_L'$ as the pivot vertex. The other is to choose a vertex $v^*$ that has some non-neighbors in $C_L\setminus C_L'$ as the pivot vertex. Symmetrically, we can also choose a pivot vertex from $C_R$. Following this direction, we adapt the existing pivoting strategy \cite{dai2023hereditary} to our setting and summarize our pivoting strategy (when $X=\emptyset$) as follows. 
% \vspace{2mm}
% \smallskip\noindent
% \fbox{
% \begin{minipage}{0.45\textwidth}
% \textbf{Pivoting Strategy ($\bm{X=\emptyset}$).} 
Let $B=(S,C,X)$ be a branch. Given a partition $C'$ and $C\setminus C'$ according to Eq.~(\ref{eq:C_L_C}), a vertex $v\in C$ is a candidate pivot if any of the following is satisfied: (1) $v\in C\setminus C'$; or (2) $v\in C_L'$ (resp. $C_R'$) has at least one non-neighbor $u\in C_R$ (resp. $C_L$) in $C_R\setminus C_R'$ (resp. $C_L\setminus C_L'$).
% \begin{enumerate}[leftmargin=*]
%     \item $v\in C\setminus C'$;
%     \item $v\in C_L'$ (resp. $C_R'$) has at least one non-neighbor $u\in C_R$ (resp. $C_L$) in $C_R\setminus C_R'$ (resp. $C_L\setminus C_L'$).
% \end{enumerate}
The pivot $v^*$ is selected with the minimum value of $\overline{d}(v^*, C_R)$ (if $v^*\in C_L$) or $\overline{d}(v^*, C_L)$ (if $v^*\in C_R$) among all candidate pivots. If there exists multiple vertices that have the same minimum value, we select the pivot arbitrarily. 
% \end{minipage}
% }
% \end{minipage}
% \vspace{2mm}
Once the pivot vertices $v^*$ is selected, the branching steps will be conducted on $v^*$ and its non-neighbors on $v^*$'s opposite side. The correctness of the pivoting strategy can be easily verified \cite{dai2023hereditary}.

% Based on this, when $X$ becomes empty, we partition the candidate set into two disjoint sets such that $C'$ always induces a 2-biplex. 
% Formally, for each $v\in C_L$, we collect the number of disconnections with the vertices in $C_R$, i.e., $\overline{d}(v, C_R)$. Symmetrically, for each $v\in C_R$, we collect the number of disconnections with the vertices in $C_L$, i.e., $\overline{d}(v, C_L)$. We define the partition $C'=C_L'\cup C_R'$ as follows, where 
% \vspace{5cm}

\smallskip\noindent\textbf{Category 2 - $\bm{X}$ is not empty.} When $X$ is non-empty, we observe that $X$ contains the vertices that have been explored in earlier branches and are typically used to check maximality. Therefore, we leverage the number of connections and/or disconnections between the vertices in $C$ and $X$ to determine our partition. 
We consider the vertices in one side of $C$ and those in the other side of $X$. Specifically, for each vertex $v\in C_L$, we collect the number of disconnections with the vertices in $X_R$, i.e., $\overline{d}(v, X_R)$. Symmetrically, we collect $\overline{d}(v, X_L)$ for each vertex $v \in C_R$. Besides, for each $v\in C$, we also collect the number of disconnections with the vertices on $v$'s opposite side. 
Based on these values, we define the partition $C' = C_L' \cup C_R'$ as follows, where
\begin{equation}
\label{eq:C_L_X}
\begin{aligned}
    C_L'= \{v\in C_L\mid \overline{d}(v,X_R)=0 \land \overline{d}(v, C_R)\le 2\} \\
    C_R'= \{v\in C_R\mid \overline{d}(v,X_L)=0 \land \overline{d}(v, C_L)\le 2\}
\end{aligned}
\end{equation}

Recall that our goal is to eliminate the presence of vertices in $X$ and the vertices in $C$ having at least 3 non-neighbors. The above definition seeks to identify a subset of vertices $C'$ from $C$ such that the branching steps produced by including the vertex $v\in C'$ will not reduce the presence of those vertices. To demonstrate this, assume without loss of generality that $v\in C_L'$.
When $X_R\ne \emptyset$, $C_L'$ consists of vertices that are fully connected to $X_R$, i.e., $N(v, X_R)=X_R$. In the case where $X_R=\emptyset$, we observe that (1) the size of $X_R$ can no longer be reduced, and (2) $\overline{d}(v, \emptyset)=0$ since there are no vertices to disconnect from. 
Thus, based on this partition, we have three options for selecting a pivot vertex. First, we can choose a vertex $v^*$ directly from $X_L$. Second, we can select $v^*$ from $C_L\setminus C_L'$, as these vertices have disconnections with some members of $X_R$. Third, we can pick $v^*$ from $C_L'$ if it has non-neighbors in $C_R\setminus C_R'$. Symmetrically, a vertex can also be chosen from $C_R\cup X_R$. We summarize our pivoting strategy (when $X\ne\emptyset$) as follows. 
% \vspace{2mm}
% \smallskip\noindent
% \fbox{
% \begin{minipage}{0.45\textwidth}
% \textbf{Pivoting Strategy ($\bm{X\ne\emptyset}$).} 
Let $B=(S,C,X)$ be a branch. Given a partition $C'$ and $C\setminus C'$ according to Eq.~(\ref{eq:C_L_X}), a vertex $v\in C\cup X$ is a candidate pivot if any of the following is satisfied: (1) $v\in (C\setminus C')\cup X$; or (2) $v\in C_L'$ (resp. $C_R'$) has at least one non-neighbor in $C_R\setminus C_R'$ (resp. $C_L\setminus C_L'$).
% \begin{enumerate}[leftmargin=*]
    % \item $v\in (C\setminus C')\cup X$;
    % \item $v\in C_L'$ (resp. $C_R'$) has at least one non-neighbor in $C_R\setminus C_R'$ (resp. $C_L\setminus C_L'$).
% \end{enumerate}
The pivot $v^*$ is selected with the minimum value of $\overline{d}(v^*, C_R)$ (if $v^*\in C_L\cup X_L$) or $\overline{d}(v^*, C_L)$ (if $v^*\in C_R\cup X_R$) among all candidate pivots. If there exists multiple vertices that have the same minimum value, we select the pivot from $X$ if possible; otherwise, we select one from $C$.  
% \end{minipage}
% }
% \end{minipage}
% \vspace{2mm}
Similarly, once the pivot vertices $v^*$ is selected, the branching steps will be conducted on $v^*$ (if $v^*\in C$) and its non-neighbors in $C$ on $v^*$'s opposite side. The correctness of the pivoting strategy can be easily verified \cite{dai2023hereditary}.

\begin{algorithm}[t]
\small
\caption{Pivot-based branch-and-bound algorithm with improved candidate set partition strategy: \texttt{IPS}}
\label{alg:impr-pivot}
\KwIn{A bipartite graph $G=(L\cup R,E)$, $\tau_L$ and $\tau_R$.}
\KwOut{All maximal bicliques in $G$ that satisfy the size constraints.}
\SetKwFunction{List}{\textsf{ParBB\_Rec}}
\SetKwProg{Fn}{Procedure}{}{}
\List{$\emptyset, L\cup R, \emptyset$}\;
\Fn{\List{$S, C, X$}} {
\tcc{Termination}
    \If{$G[S\cup C]$ induces a 2-biplex and $X=\emptyset$}{
        \textbf{Output} all maximal bicliques in $G[S\cup C]$\;
        \Return\;
    }
\tcc{Pruning}
    \If{any of pruning rules is satisfied}{\Return\;}
\tcc{Partition-based Branching with Pivoting}
    % \If {$X=\emptyset$} {
    %     $C'\gets$ vertices that satisfy Eq.~(\ref{eq:C_L_C})\;
    %     $v^*\gets$ the pivot selected by Pivoting Strategy ($X=\emptyset$)\;
    %     Create branches $B_i$ on $v^*$ and its non-neighbors on $v^*$'s opposite side\;
    %     % $v^*\gets$ a vertex in $(C\setminus C')\cup X$ with the min. value of $\overline{d}(\cdot, C\setminus C')$\;
    %     % \lIf{$v^*\in X$}{Create $\overline{d}(v^*, C\setminus C')$ branches over the vertices in $\overline{N}(v^*, C\setminus C')$ based on Lemma~\ref{lemma:basic-pivot}}
    %     % \lElse{Create $\overline{d}(v^*, C\setminus C') + 1$ branches over the vertices in $\{v^*\}\cup \overline{N}(v^*, C\setminus C')$ based on Lemma~\ref{lemma:basic-pivot}}
        
    % }
    % \Else{
        % $C'\gets$ vertices that satisfy Eq.~(\ref{eq:C_L_X})\;
        $v^*\gets$ the pivot selected by Partition-based Pivoting Strategy\;
        Create branches $B_i$ on $v^*$ (if $v^*\in C$) and its non-neighbors on $v^*$'s opposite side\;
        % $v^*\gets$ a vertex in $C\setminus C'$ with the min. value of $\overline{d}(\cdot, C)$\;

        % Create $\overline{d}(v^*, C) + 1$ branches over the vertices in $\{v^*\}\cup \overline{N}(v^*, C)$ based on Lemma~\ref{lemma:basic-pivot}
    % }
    
    \For{each created branch $B_i$}{\List{$S_i, C_i, X_i$}\;}
}
\end{algorithm}

\subsection{Implementation Details and Analysis}
\label{subsec:summary}

Before we introduce the implementation details of the algorithm, we first elaborate on the possible pruning rules that can be applied as follows.

\vspace{2mm}
\noindent
\fbox{
\begin{minipage}{0.45\textwidth}
\textbf{Pruning Rules.} Let $B=(S,C,X)$ be a branch. Define two variables $\theta_L=|S_L|+|C_L|$ and $\theta_R=|S_R|+|C_R|$. We can prune the branch $B$ if any of the following is satisfied:
\begin{enumerate}
    \item[(P1)]  $\theta_L<\tau_L$ or $\theta_R < \tau_R$;
    \item[(P2)]  There exists a vertex $v\in X_L$ (resp. $X_R$) such that $v$ connects with all vertices in $C_R$ (resp. $C_L$).
\end{enumerate}
\end{minipage}
}
% \end{minipage}
\vspace{2mm}

The rationale of the first rule is that $\theta_L$ and $\theta_R$ can be verified to be the upper bound of the number of vertices at the left side and that at the right side of a maximal biclique covered by the branch $B$, respectively. If either of them does not comply with the size constraint, the maximal biclique enumerated from this branch would violate the problem definition. 
The rationale of the second rule is that all maximal bicliques in subgraph $G[S\cup C]$ would not be globally maximal, since an additional vertex $v$ can be included in each of them. 
In addition, we observe that the partition strategies and the pivoting strategies can be naturally combined when $X=\emptyset$ and $X\ne \emptyset$. We summarize as follows.

\vspace{2mm}
\noindent
\fbox{
\begin{minipage}{0.45\textwidth}
\textbf{Partition-based Pivoting Strategy.} Let $B=(S,C,X)$ be an arbitrary branch during the recursion. We partition $C$ into two disjoint sets $C'=C_L'\cup C_R'$ and $C\setminus C'$ by setting
\begin{equation}
\label{eq:pivot}
\begin{aligned}
    C_L'= \{v\in C_L\mid \overline{d}(v,X_R)=0 \land \overline{d}(v, C_R)\le 2\} \\
    C_R'= \{v\in C_R\mid \overline{d}(v,X_L)=0 \land \overline{d}(v, C_L)\le 2\}.
\end{aligned}
\end{equation}
A vertex $v$ is a candidate pivot if any of the following is satisfied: 
% (1) $v\in (C\setminus C')\cup X$; or (2) $v\in C_L'$ (resp. $C_R'$) has at least one non-neighbor in $C_R\setminus C_R'$ (resp. $C_L\setminus C_L'$).
\begin{enumerate}
    \item[(C1)] $v\in (C\setminus C')\cup X$;
    \item[(C2)] $v\in C_L'$ (resp. $C_R'$) and it has at least one non-neighbor in $C_R\setminus C_R'$ (resp. $C_L\setminus C_L'$).
\end{enumerate}
The pivot $v^*$ is selected among all candidates with the minimum number of non-neighbors in $C$. We first select the pivot from $X$ if possible; otherwise, we select one from $C$.
\end{minipage}
}
\vspace{2mm}

Based on the all possible techniques, including (1) the partition-based framework, (2) the candidate set partition strategies, (3)  the pivoting strategy and (4) pruning techniques, we develop a new branch-and-bound algorithm called \texttt{IPS}. The pseudocode is presented in Algorithm~\ref{alg:impr-pivot}.

\smallskip\noindent
\textbf{Worst-case Time Complexity Analysis.}
The overall time complexity can be expressed as $O\bigl(P(n)\cdot T(n)\bigr)$, where $P(n)$ denotes the running time per recursive call and $T(n)$ denotes the total number of recursive calls. 
For the former, given a branch $B=(S,C,X)$, we note that this new partition-based pivoting strategy will not incur more overhead than that of the existing pivoting strategy, i.e., the partition-based pivoting strategy can also select a pivot $v^*$ in $O(m)$ time. To see this, our pivot selection procedure can be decomposed into three steps. The first step is to calculate necessary values for partition, i.e., $\overline{d}(\cdot, \cdot)$. The second step is to collect the candidate pivots based on the partition. The last step is to choose the pivot among candidates. Consider the first step. We take a vertex $v\in C_L\cup X_L$ as an example. To calculate $\overline{d}(v, X_R)$ (resp. $\overline{d}(v, C_R)$), we first initialize $\overline{d}(v, X_R)$ (resp. $\overline{d}(v, C_R)$) to be $|X_R|$ (resp. $|C_R|$). Then we substract one from $\overline{d}(v, X_R)$ (resp. $\overline{d}(v, C_R)$) if there exists an edge $(v, u)$ between $v$ and a vertex $u\in X_R$ (resp. $C_R$), which can be done in $O(m)$. Symmetric procedure can be done for a vertex $v\in C_R\cup X_R$. Consider the second step. It can be done in $O(|C|)$, which is bounded by $O(m)$. Consider the third step. The critical point is to calculate a vertex $v$'s non-neighbors. We only collect the non-neighbor information for the vertices that have at most two non-neighbors, which can be done in $O(m+2\cdot |C|)$. Thus, $P(n)$ is bounded by $O(m)$. 
For the latter, we present the detailed analysis of $T(n)$ and the resulting time complexity in the following theorem, which shows that the worst-case bound of \texttt{IPS} is strictly better than that of existing algorithms.

\begin{theorem}
\label{theo:ps+}
Given a bipartite graph $G=(L\cup R, E)$ with $n=|L|+|R|$ and $m=|E|$, \texttt{IPS} enumerates all maximal bicliques in $G$ in $O(m \cdot \alpha^n + n\cdot \beta)$ time, where $\alpha~(\approx 1.3954)$ is the largest positive real root of $x^4-2x-1=0$. 
\end{theorem}

{\Revise
\begin{proof}[Proof]
Let $B=(S,C,X)$ be a branch and let $n=|C|+|X|$. 
The running time of \textsf{ParBB\_Rec} is dominated by the size of its recursion tree, i.e., the total number of generated branches. 
At each branch, the pivot $v^*$ is selected and the algorithm creates sub-branches on $v^*$ (if $v^*\in C$) and on its $k$ non-neighbors on the opposite side, where $k=\overline{d}(v^*,C)$. 
Let $T(n_0)$ denote the number of recursive calls involving $v^*$ (when applicable), and let $T(n_1),\cdots,T(n_k)$ denote the number of recursive calls involving each non-neighbor of $v^*$, respectively. Accordingly, the number of recursive calls satisfies
\begin{equation}
T(n)\ \le\
\begin{cases}
\sum_{i=1}^{k} T(n_i) &\quad\quad \text{if } v^*\in X, \\
T(n_0) + \sum_{i=1}^{k} T(n_i) &\quad\quad \text{if } v^*\in C,
\end{cases}
\end{equation}
where each $n_i$ denotes the size of the remaining instance in the corresponding sub-branch, determined by the number of vertices removed from $C\cup X$. Each branching step yields a linear recurrence, and the exponential growth rate of its solution is determined by the largest real root of the corresponding characteristic equation\footnote{For a linear recurrence $T(n)\le \sum_{i} c_i\cdot T(n-a_i)$, one assumes a solution of the form $T(n)=x^n$, which leads to the characteristic equation $x^n - \sum_i c_i \cdot x^{n- a_i} = 0$. The largest real root of this equation determines the asymptotic growth rate of $T(n)$~\cite{kratsch2010exact}.}.

The analysis systematically characterizes all possible branching cases. Without loss of generality, assume the pivot $v^*$ is selected from the left side, i.e., $v^*\in C_L\cup X_L$. Depending on the specific position where $v^*$ is selected, each branch falls into one of the following three categories. For each category, we derive the recurrence by tracking the reduction of $C\cup X$ in each sub-branch and upper bound the total number of recursive calls accordingly.

\smallskip\noindent
\underline{\textbf{Case (1) - $\bm{v^*\in X_L}$.}}
If $k=0$, the branch is pruned by Rule (P2); hence $k\ge 1$. The algorithm creates $k$ sub-branches on $v^*$'s non-neighbors in $C_R$. It can be verified that each sub-branch removes $v^*$, the chosen non-neighbor, and at least $k$ additional vertices on the opposite side, giving a size decrease of at least $k+2$.
Therefore,
\begin{equation}
T(n)\ \le\ k\cdot T\bigl(n-(k+2)\bigr)\qquad (k\ge 1).    
\end{equation}
This recurrence implies $T(n)=O(n)$ when $k=1$, and an exponential bound with base $\alpha_1=\max_{k\ge 2} k^{1/(k+2)}<1.26$ when $k\ge 2$.

\smallskip\noindent
\underline{\textbf{Case (2) - $\bm{v^*\in C_L \setminus C_L'}$.}} Given $v^*\notin C_L'$, either $\overline{d}(v^*, X_R) \neq 0$ or $\overline{d}(v^*, X_R) = 0 \land \overline{d}(v^*, C_R) > 2$ holds. We consider both cases.

\smallskip\noindent\textbf{Case (2.1) - $\bm{\overline{d}(v^*, X_R)\ne 0}$.}
It follows that (i) $v^*$ has at least one non-neighbor in $X_R$ and (ii) all non-neighbors of $v^*$ in $C_R$ are candidate pivots. 
The algorithm branches on $v^*$ and on its $k$ non-neighbors in $C_R$.
The branching step on $v^*$ removes $v^*$, all $k$ non-neighbors in $C_R$, and one excluded non-neighbor in $X_R$, resulting in a decrease of at least $k+2$. 
Each branching step on a non-neighbor of $v^*$ removes the chosen vertex and, due to the pivot selection rule, enforces the removal of at least $k$ additional vertices on the opposite side, resulting in a decrease of at least $k+1$. 
Thus,
\begin{equation}
T(n)\ \le\ T\bigl(n-(k+2)\bigr)\ +\ k\cdot T\bigl(n-(k+1)\bigr)\qquad (k\ge 0).
\end{equation}
The worst case occurs at $k=2$, giving the exponential bound with base $\alpha_2\approx 1.3954$, where $\alpha_2$ is the largest real root of $x^4-2x-1=0$.

\smallskip\noindent
\textbf{Case (2.2) - $\bm{\overline{d}(v^*, X_R)= 0}$.}
It follows that (i) $k\ge 3$ and (ii) all non-neighbors of $v^*$ in $C_R$ are candidate pivots. The algorithm creates $k+1$ sub-branches on $v^*$ and its $k$ non-neighbors in $C_R$. 
By the pivot selection rule, branching on either $v^*$ or its non-neighbors removes at least $k+1$ vertices from $C\cup X$, yielding a decrease of at least $k+1$ in each sub-branch. This gives the coarse bound
\begin{equation}
T(n)\ \le\ (k+1)\cdot T\bigl(n-(k+1)\bigr) \qquad (k\ge 3).
\end{equation}
The worst case occurs at $k=3$. To achieve a tighter bound, a finer-grained analysis at $k=3$ leads to the following recurrence
\begin{equation}
T(n)\ \le\ 2\cdot T(n-4)\ +\ 2\cdot T(n-6)\ +\ 2\cdot T(n-7),
\end{equation}
whose branching factor is strictly smaller than $\alpha_2$. 
For $k\ge 4$, the coarse bound applies, with exponential base $\max_{k\ge 4}(k+1)^{1/(k+1)} < 1.38< \alpha_2$.
We put the detailed analysis for $k=3$ as follows.

The algorithm produces four sub-branches, one on $v^*$ and one on each of its non-neighbors $u_1$, $u_2$ and $u_3$, denoted by $B_0$, $B_1$, $B_2$ and $B_3$, respectively. Consider $B_i=(S_i, C_i, X_i)$ produced on $u_i$ for $1\le i\le 3$. It can be verified that $u_i\in C_R \setminus C_R'$. Let $n_1$, $n_2$ and $n_3$ be the size of $|C_i|+|X_i|$ for $1\le i\le 3$, respectively. We first have $n_i = |C_i|+|X_i|\le |C|+|X|-4=n-4$. 
% Note that the worst case appears when all the produced branches have exactly $n-4$ vertices in $C_i\cup X_i$. Thus, we provide a finer-grained analysis for the worst case as follows.  
We skip the analysis for $B_1$ since $n_1$ cannot achieve a tighter bound. Consider the branch $B_2=(S_2,C_2,X_2)$ produced on $u_2$. Note that $u_1\in X_2$ and $v^*\notin C_2\cup X_2$. Since $\overline{d}(u_1, C_L)=3$ under the worst case, we have $\overline{d}(u_1, {C_2}_L)=2$, where ${C_2}_L$ is the left side of $C_2$. Then consider the pivot selected in the branch $B_2$, which we denote by $v_2^*$. Either of the following case holds: (i) if $v_2^*\in X_2$, then $\overline{d}(v_2^*, C_2)\le 2$ or (ii) if $v_2^*\in C_2$, then $\overline{d}(v_2^*, C_2)\le 1$; since otherwise $u_1$ will be the pivot. 

For the former case, two sub-branches would be produced from $B_2$ under the worst case, i.e., those produced upon $v_2^*$'s two non-neighbors in $C_2$. We have
    \begin{equation}
        T(n_2)\le 2\cdot T(n_2-4) \le 2\cdot T(n-8)
    \end{equation}
    
For the latter case, two sub-branches would be produced from $B_3$ under the worst case, i.e., one produced upon $v_2^*$ and the other produced upon $v_2^*$'s non-neighbor in $C_2$. We have 
    % the recurrence.
    \begin{equation}
        T(n_2)\le T(n_2-2) + T(n_2-3) \le T(n-6) + T(n-7)
    \end{equation}
    
It is easy to see that the worst case would happen in the latter one. Similar to the analysis of $B_2$, we can have the same recurrence for the sub-branch $B_3$. Thus, we omit the details of the recurrence for $T(n_3)$. To sum up, we have 
    \begin{equation}
        T(n) \le 2 \cdot T(n-4) + 2 \cdot T(n-6) + 2 \cdot T(n-7)
    \end{equation}

\smallskip\noindent
\underline{\textbf{Case (3) - $\bm{v^*\in C_L'}$.}} It follows $k\in\{1,2\}$. We analyze both cases.

\smallskip\noindent\textbf{Case (3.1) - $k=1$.} 
The algorithm produces two sub-branches, one on $v^*$ and the other on its non-neighbor $u_1$, denoted by $B_0$ and $B_1$, respectively.
For $B_0$, the branching step removes $v^*$ and $u_1$, and we have $|C_0|+|X_0|\le (|C|-2)+|X|=n-2$. 
For $B_1$, since $u_1\in C_R\setminus C_R'$ by Rule (C2), either $\overline{d}(u_1, C_L)>2$ or $\overline{d}(u_1, X_L)>0$ holds, and in both cases, $|C_1|+|X_1|\le |C|+|X|-3= n-3$. 
Therefore,
\begin{equation}
T(n) \le T(n-2) + T(n-3).
\end{equation}

\smallskip\noindent\textbf{Case (3.2) - ${k=2}$.}
The algorithm produces three sub-branches, one on $v^*$ and one on each of its non-neighbors $u_1$ and $u_2$, denoted by $B_0$, $B_1$, and $B_2$, respectively. 
By Rule (C2), we assume $u_2\in C_R\setminus C_R'$. Let $l=\overline{d}(u_1, C_L)\ge 1$, since $u_1$ is non-adjacent to at least one vertex, i.e., $v^*$. We distinguish the following cases based on the value of $l$.

 \smallskip\noindent\textbf{Case (3.2.1) - $l=1$.}
We claim that no maximal biclique exists in $B_2$ since (i) $u_1\in X_2$ and (ii) for any biclique found in $B_2$, the vertex $u_1$ can always be added to form a strictly larger biclique, contradicting maximality.
Therefore, the sub-branch $B_2$ can be safely pruned. 
When $k=2$ and $l=1$, we have $|C_0|+|X_0|\le (|C|-3)+|X|=n-3$ and $|C_1|+|X_1|\le (|C|-2)+|X|=n-2$. 
Hence,
\begin{equation}
T(n) \le T(n-2) + T(n-3).
\end{equation}

\smallskip\noindent\textbf{Case (3.2.2) - $l\ge 2$.} Consider $B_0$. Given $k=2$, we have $|C_0|+|X_0|\le (|C|-3)+|X|=n-3$. Consider $B_1$. Given $l\ge 2$, we have $|C_1|+|X_1|\le (|C|-3)+|X|=n-3$. Consider $B_2$. $u_2$ is a candidate pivot since $u_2\in C_R\setminus C_R'$. Thus, either (i) $\overline{d}(u_2, C_L)>2$ or (ii) $\overline{d}(u_2, C_L)=2 \land \overline{d}(u_2, X_L)>0$ holds. Then $|C_2|+|X_2|\le |C|+|X|-4= n-4$. Formally, we have 
    \begin{equation}
        T(n) \le 2\cdot T(n-3) + T(n-4)
    \end{equation}

By solving all recurrences, the worst case is attained in Case (3.2.2), whose branching factor $\alpha_3$ is the same as that of Case (2.1).

\smallskip
\noindent\textbf{Conclusion.}
Let $\alpha=\max\{\alpha_1,\alpha_2,\alpha_3\}$.
Then the total number of recursive calls is bounded by $O(\alpha^n)$. Since each call can be done in $O(m)$ time, all maximal bicliques can be enumerated in $O(m\cdot \alpha^n + n\cdot \beta)$ time, where $\beta$ is the number of maximal bicliques.
\end{proof}
}

\begin{figure*}[t]
\centering
% \subfigure[\textsf{CR}]{
%     \includegraphics[width=0.24\textwidth]{figs/empty.png}}
% \subfigure[\textsf{FR}]{
%     \includegraphics[width=0.24\textwidth]{figs/empty.png}}
% \subfigure[\textsf{WK}]{
%     \includegraphics[width=0.24\textwidth]{figs/empty.png}}
% \subfigure[\textsf{ES}]{
%     \includegraphics[width=0.24\textwidth]{figs/empty.png}}
\subfigure[A 3-biplex graph $G$]{
\label{subfig:example}
    \includegraphics[width=0.3\textwidth]{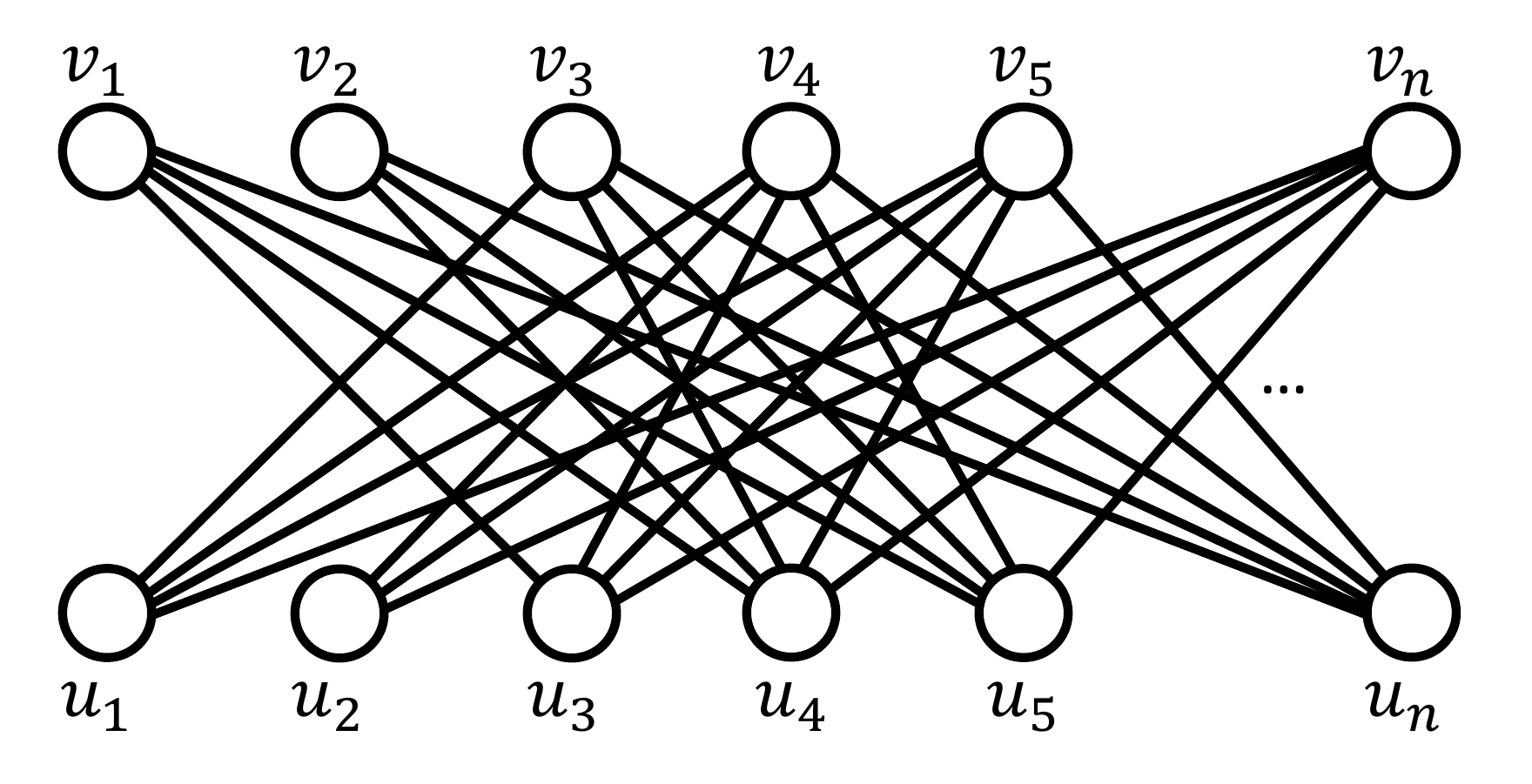}}
\hspace{0.5cm}
\subfigure[$G$'s complementary graph  $\overline{G}$]{
\label{subfig:example2}
    \includegraphics[width=0.3\textwidth]{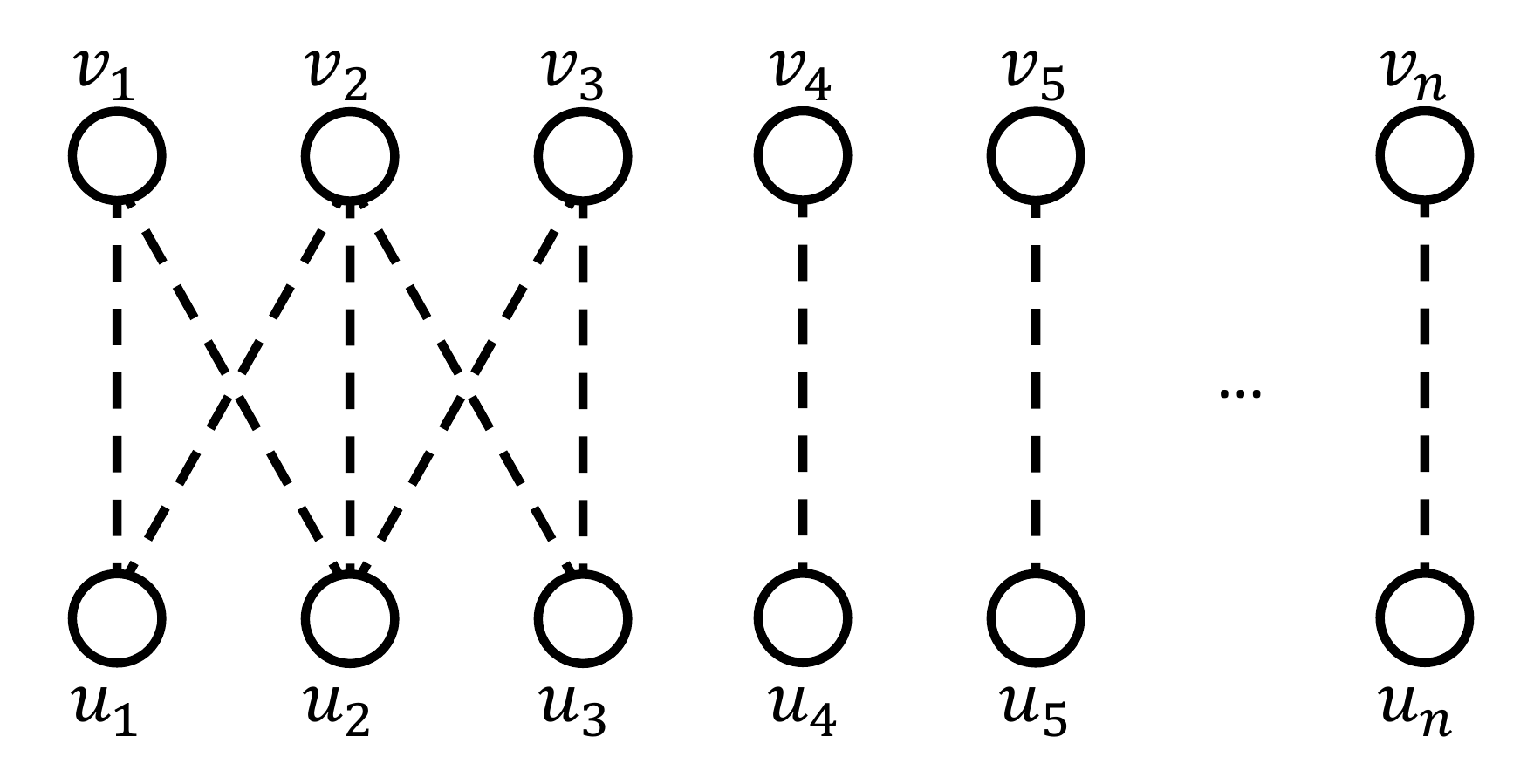}}
\hspace{0.5cm} 
\subfigure[The complementary of the candidate graphs after branching on $u_4$ and/or $v_4$ with the basic pivot strategy, i.e., $u_4$ is the pivot vertex.]{
\label{subfig:branch1}
    \includegraphics[width=0.255\textwidth]{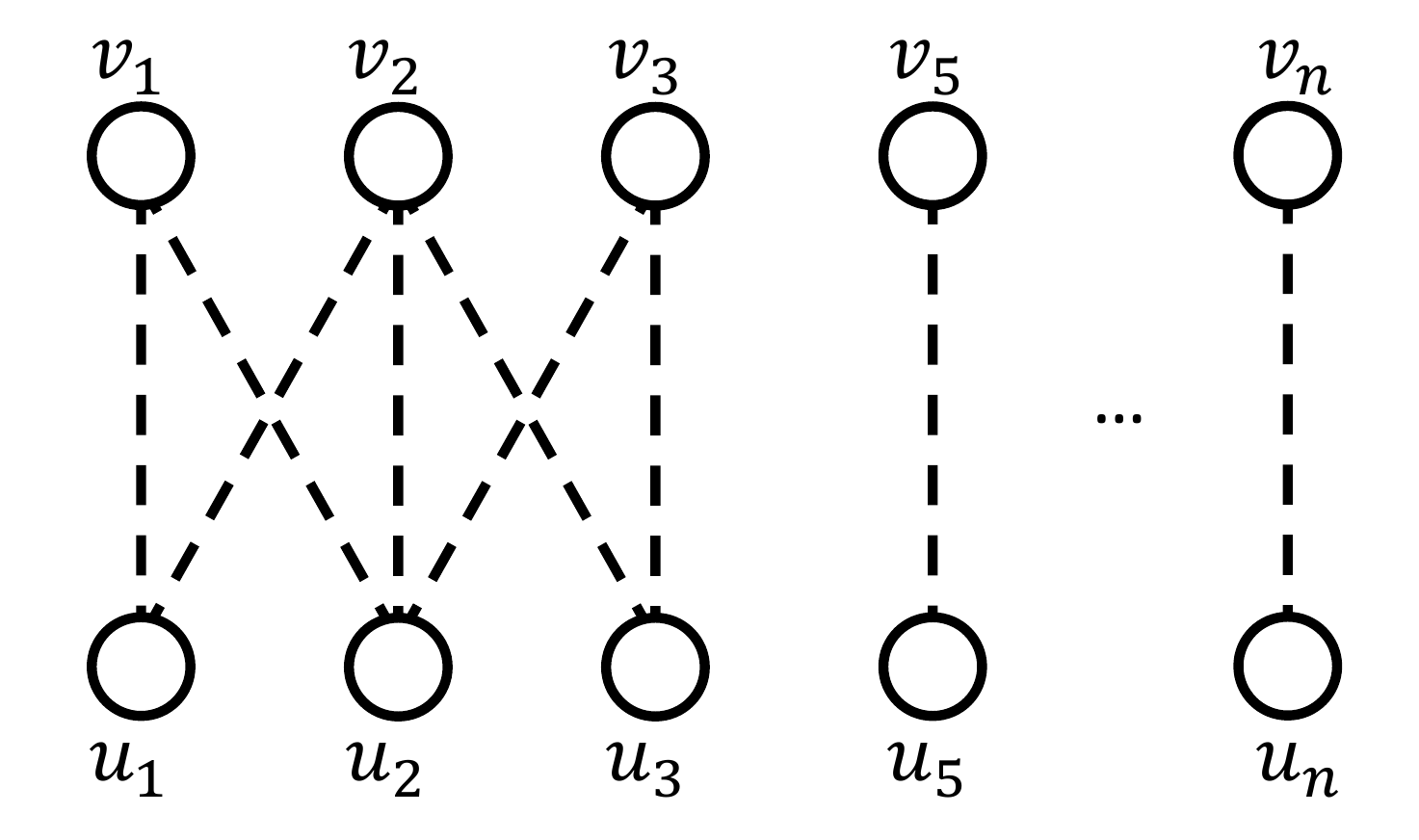}}
\\
\subfigure[The complementary of the candidate graph after branching on $u_1$ with another pivot strategy, e.g., $u_1$ is the pivot vertex.]{
\label{subfig:branch2}
    \includegraphics[width=0.255\textwidth]{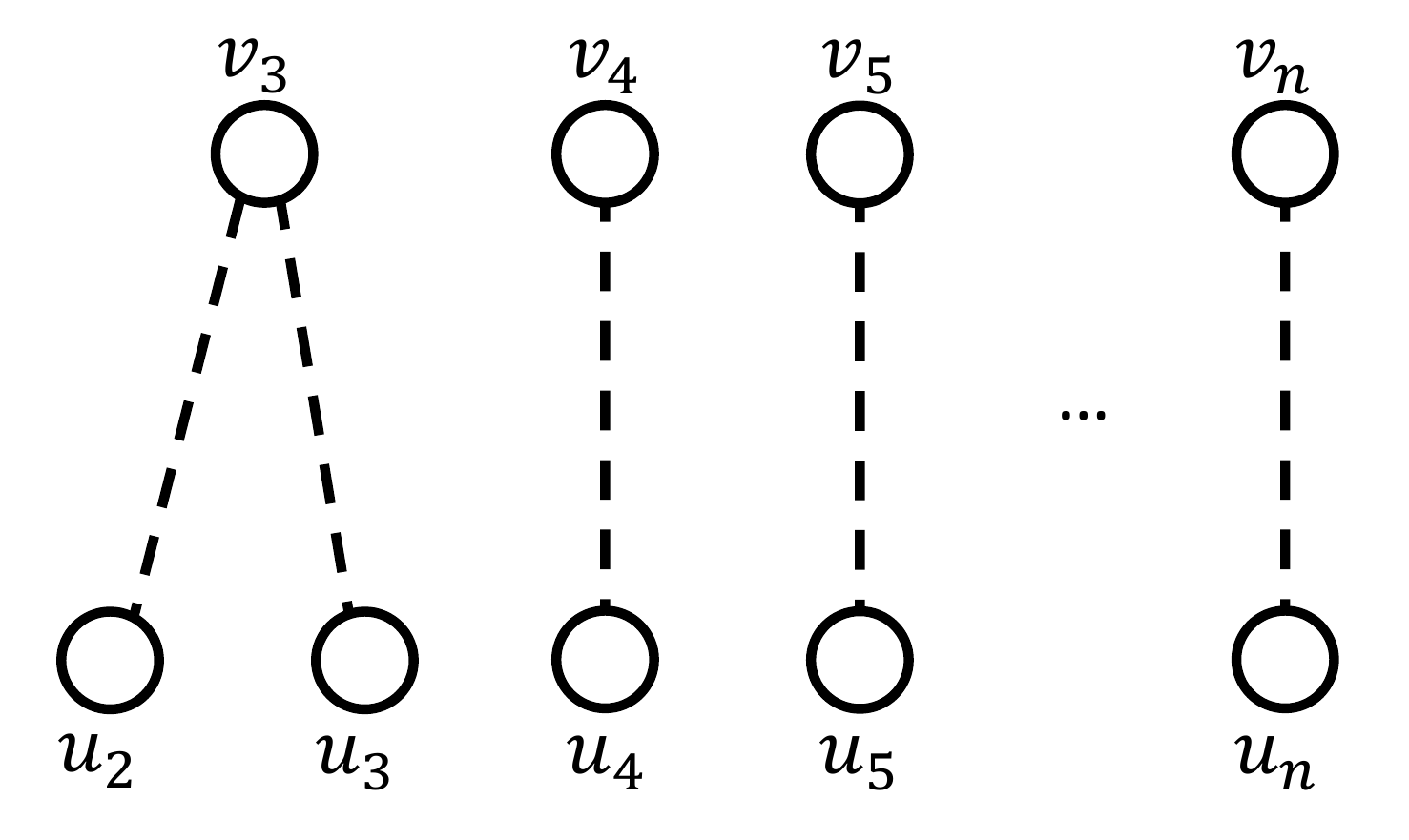}}
\hspace{0.8cm}
\subfigure[The complementary of the candidate graph after branching on $v_1$ with another pivot strategy, e.g., $u_1$ is the pivot vertex.]{
\label{subfig:branch3}
    \includegraphics[width=0.255\textwidth]{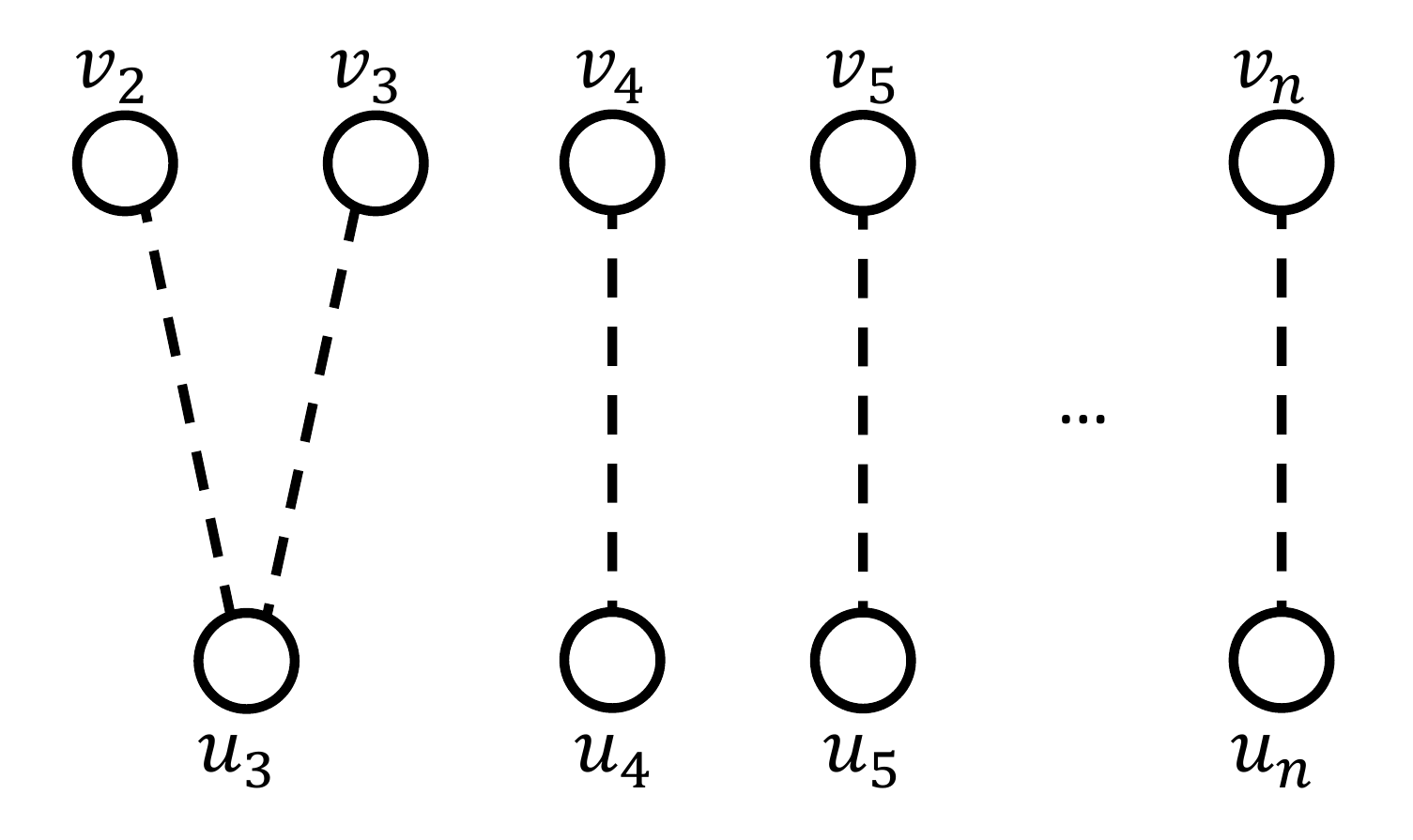}}
\hspace{0.8cm}
\subfigure[The complementary of the candidate graph after branching on $v_2$ with another pivot strategy, e.g., $u_1$ is the pivot vertex.]{
\label{subfig:branch4}
    \includegraphics[width=0.255\textwidth]{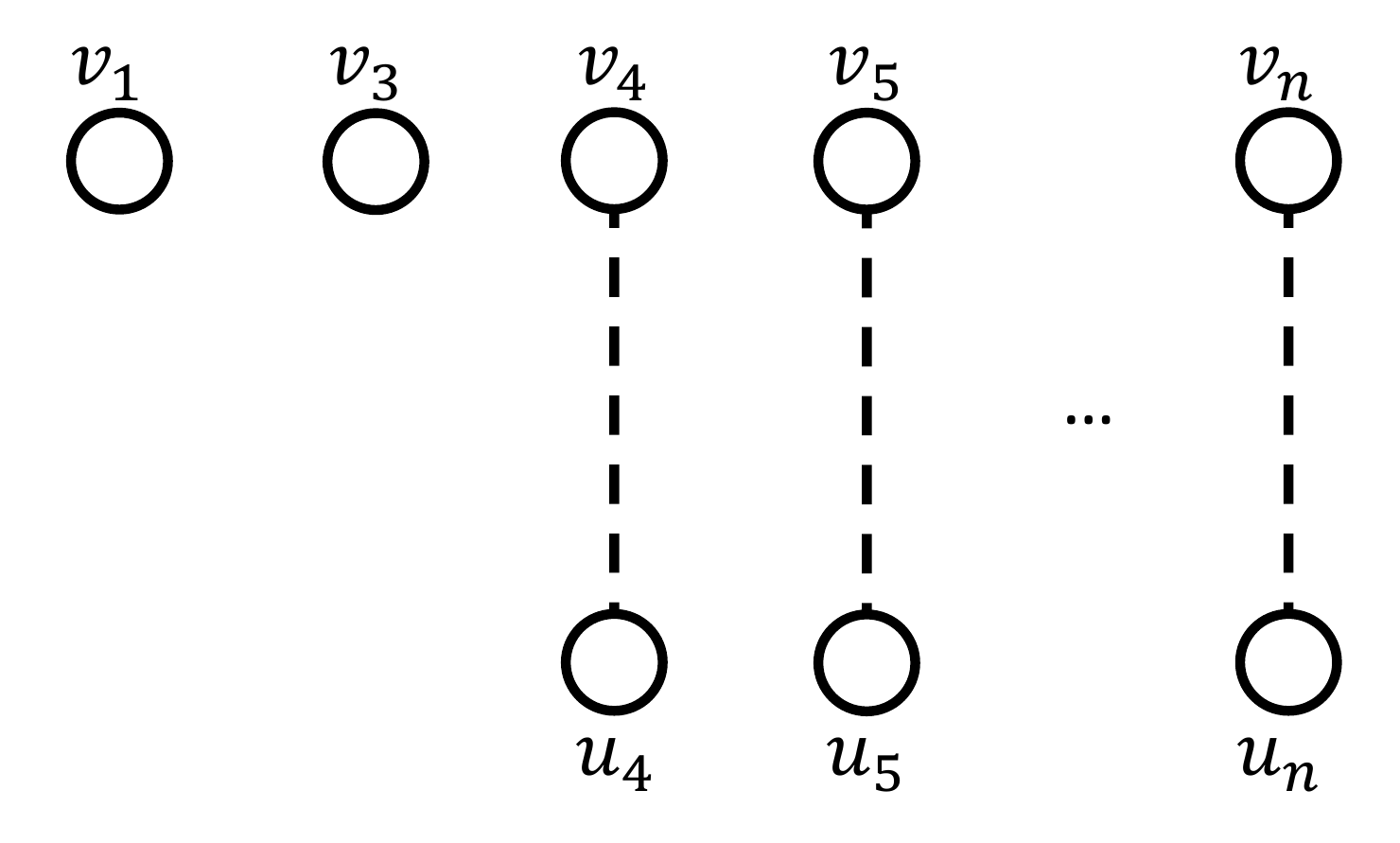}}
%\vspace{-3mm}
\caption{An illustration of the reason why the conventional pivoting strategy with the new stopping criterion would still produce much more branches than necessary.}
%\vspace{-3mm}
\label{fig:example}
\end{figure*}

\smallskip\noindent
\textbf{Remarks.} 
\underline{First}, compared with existing algorithms that target the maximal biclique enumeration problem \cite{chen2022efficient, dai2023hereditary, abidi2020pivot, chen2023maximal}, our algorithm offers several key advantages. Theoretically, while existing branch-and-bound algorithms \cite{chen2022efficient, dai2023hereditary, abidi2020pivot} all exhibit a worst-case time complexity of $O(m \cdot (\sqrt{2})^n)$, our approach provides a strictly better theoretical bound since both $O(m\cdot \alpha^n)$ and $O(n\cdot \beta)$ are bounded by $O(m\cdot (\sqrt{2})^n)$. We note that $\beta$ is at most $(\sqrt{2})^n - 2$ in the worst case \cite{prisner2000bicliques}, yet empirically, it is much smaller. Practically, our algorithm outperforms all existing algorithms, as confirmed by the experimental results.

{\Revise We emphasize that the termination condition and the pivoting strategy \emph{jointly} contribute to the theoretical improvement in the worst-case performance. Although each technique independently contributes to the overall empirical performance, neither is sufficient alone to achieve a substantial reduction in complexity. 
We also acknowledge that the excellent empirical performance of state-of-the-art algorithms and ours on real-world graphs primarily arises from the structural properties of such datasets (e.g., sparsity), rather than from the worst-case behavior. However, measuring such average-case complexity is challenging and even infeasible in many instances, which is a common issue in the field. Thus, identifying algorithmic bottlenecks and providing a solid foundation for future improvements is still essential.
}

\underline{Second}, recall that in Section~\ref{sec:intro}, we can directly build upon the existing BB algorithm~\cite{dai2023hereditary} by just integrating the our new stopping criterion with the existing BB algorithms and retaining the conventional pivoting strategy. However, the conventional pivoting strategy is not well-aligned with the new stopping criterion and would still generate much more branches than necessary, i.e., $O((\sqrt{2})^n)$ branches. Therefore, this design still has the same worst-case time complexity as that of \cite{dai2023hereditary}. We illustrate this with the following example.

\begin{example}
    Consider a bipartite graph $G$ in Figure~\ref{subfig:example} whose complementary graph $\overline{G}$ is shown in Figure~\ref{subfig:example2}. We cannot directly apply the termination technique to $G$ since $G$ is a 3-biplex (i.e., $u_2$ and $v_2$ have 3 non-neighbors on their opposite side, respectively). Therefore, the conventional pivoting strategy would choose a vertex with the minimum number of non-neighbors as the pivot, e.g., $u_4$ is selected as the pivot vertex. Branching steps would be conducted on $u_4$ and its non-neighbor $v_4$. We show the complementary graph of the candidate graphs after branching on $u_4$ and/or $v_4$ in Figure~\ref{subfig:branch1}. We note that these two candidate graphs are identical. We still cannot apply the early-termination technique on either of them. The procedure would continue recursively, choosing $u_5, u_6, \cdots, u_n$ as pivots, until the candidate graph only has $u_1, u_2, u_3, v_1, v_2$ and $v_3$ inside. Then we can apply the early-termination technique on the candidate graphs after the branching steps on either of a vertex and its non-neighbors. Thus, the algorithm with the conventional pivoting strategy would produce $O(2^n)$ branches in this example graph. However, our proposed pivoting strategy, would create $O(1)$ branches, regardless of how large the value of $n$ is. That is, we skip all those $u_i$ and $v_i$ for $i\ge 4$ and choose $u_1$ as the pivot vertex. Branching steps would be conducted on $u_1$ and its non-neighbors $v_1$ and $v_2$. We show the complementary of the candidate graphs after branching on $u_1$, $v_1$ and $v_2$ in Figure~\ref{subfig:branch2}, Figure~\ref{subfig:branch3} and Figure~\ref{subfig:branch4}, respectively. Since all these three branches are 2-biplexes, we can directly apply the early-termination technique and enumerate the maximal bicliques inside in nearly optimal time. 
\end{example}

\if 0
\underline{Third}, we notice that the delay, i.e., the longest running time between any two consecutive outputs, is often used for evaluating enumeration algorithms \cite{zhang2014finding, yu2022efficient}. 
Specifically, there exist several algorithms that could provide a guarantee of polynomial delay for maximal biclique enumeration problem \cite{dai2023hereditary, chen2022efficient}. 
Our algorithm, unfortunately, cannot provide such guarantee, i.e., the delay could be exponential with respect to the size of the bipartite graph. The reason is that the early-termination would output the maximal bicliques in a batch manner, and the delay between two consecutive batches is exponential due to the backtracking nature of our algorithm. Besides, we remark that all algorithms cannot achieve the polynomial delay under the setting of non-trivial size constraints, i.e., $\tau_L>1$ and/or $\tau_R>1$.  
\fi

{\Revise \underline{Third}, the delay, i.e., the maximum running time between two consecutive outputs, is a standard measure for evaluating enumeration algorithms \cite{zhang2014finding, yu2022efficient}. 
In the literature on maximal subgraph enumeration, including the problem studied in this paper, most state-of-the-art solutions adopt a backtracking (a.k.a.\ branch-and-bound) framework \cite{wang2025maximal, abidi2020pivot, chen2023maximal}. 
This framework is widely used because it runs much faster on real-world graphs and does not require expensive overheads for maintaining complex auxiliary data structures.
Our algorithm follows this paradigm and therefore enjoys the same characteristics. As a result, it cannot guarantee polynomial delay in theory. 
In particular, the early-termination mechanism outputs maximal bicliques in batches, and the delay between two consecutive batches can be exponential in the worst case. 
Moreover, under non-trivial size constraints (i.e., $\tau_L>1$ and/or $\tau_R>1$), no existing algorithm can achieve polynomial delay in the worst case.
Although polynomial-delay algorithms for maximal subgraph enumeration do exist \cite{dai2023hereditary, chen2022efficient}, such methods typically incur more overhead and are rarely competitive in practice.

}

\underline{Finally}, our partition-based pivoting strategy can also be applied to other enumeration task, such as maximal clique enumeration. The idea of the pivot selection is similar but the difference comes from the fact that the graph is unipartite in maximal clique enumeration. Therefore, it would not be necessary to distinguish which side of a vertex and its non-neighbors. Similarly, the theoretical result can also be achieved for maximal clique enumeration. We shall leave this as a future work.

\section{Efficiency Boosting Techniques}
\label{sec:boost}

In this section, we apply a widely-used decomposition technique, namely \emph{inclusion-exclusion based framework} \cite{chen2021efficient, conte2018d2k, yu2023maximum} (\texttt{IE} in short) to further boosts the efficiency and scalability of the algorithms introduced in the Section~\ref{sec:BB-v2}.

% \subsection{Inclusion-exclusion based framework}

Given a subgraph enumeration task on a graph $G$, the idea of the framework is that it first decomposes the whole graph into multiple smaller subgraphs, then the algorithm runs on each graph instance, and finally returns the answer by considering the results from each graph instance. For our problem, given a graph $G=(L\cup R, E)$, we consider an arbitrary ordering over the vertices in $L$, which we denote by $\langle v_1, v_2, \cdots, v_{|L|}\rangle$. 
{\Revise Let $N_1(v_i)$ and $N_2(v_i)$ denote the sets of 1-hop and 2-hop neighbors of $v_i$, respectively. All maximal bicliques in $G$ can be enumerated exactly once from each induced graph $G_i = G[L_i \cup R_i]$ for $i = 1, \dots, |L|$, where $L_i$ is the set of $v_i$'s 2-hop neighbors, excluding $\{v_1, \dots, v_{i-1}\}$, and $R_i$ is the set of $v_i$'s 1-hop neighbors.}
% $L_i = N_2(v_i) \setminus \{v_1, \cdots, v_{i-1}\}$, $R_i =  N_1(v_i)$ and $E_i = \{(v, u) \mid v\in L_i \land u\in R_i \land (v, u)\in E\}$ for $1\le i \le |L|$.
% \begin{equation}
% \begin{aligned}
%     L_i &= N_2(v_i) \setminus \{v_1, \cdots, v_{i-1}\} \quad
%     R_i =  N(v_i, R)\\
%     E_i &= \{(v, u) \mid v\in L_i \land u\in R_i \land (v, u)\in E\} \\
% \end{aligned}
% \end{equation}
% where $N_2(v_i)$ denotes the set of all $v_i$'s 2-hop neighbors. 
% Note that all $v_i$'s 2-hop neighbors must reside in $L$ since $G$ is a bipartite graph and $v_i\in L$. 
The correctness of the formulation can be found in  \cite{chen2022efficient}. 
With this formulation, our branch-and-bound algorithms can solve each graph instance if the sets $S$, $C$ and $X$ are initialized properly. Formally, for the graph instance $G_i$ for $1\le i\le |L|$, 
\begin{equation}
\begin{aligned}
    \label{eq:cx}
    S_i = \{v_i\} \quad&\quad C_i = (N_1(v_i)\cup N_2(v_i)) \setminus \{v_1, \cdots, v_i\} \\
    X_i &= N_2(v_i) \cap \{v_1, \cdots, v_{i-1}\}
\end{aligned}
\end{equation}

We note that with this framework, all algorithms introduced in Section~\ref{sec:BB-v2} can be improved. The reason is as follows. {\Revise It can be verified that $C_i$ and $X_i$ are mutually exclusive, and their union is given by $(N_1(v_i) \cup N_2(v_i)) \setminus \{v_i\}$. We define a variable
\begin{equation}
\label{eq:gamma}
    \gamma = \max_{i} \bigl(|C_i| + |X_i|\bigr) = \max_{v_i \in L} \bigl(|N_1(v_i)| + |N_2(v_i)| - 1\bigr),
\end{equation}
which represents the size of the largest instance produced from the root branch.}
% $|C_i|+|X_i|$ contains all $v_i$'s 1-hop and 2-hop neighbors excluding itself. 
% We define $\gamma = \max_{i} (|C_i|+|X_i|)$. 
Theoretically, $\gamma$ is strictly smaller than $n$ since $\gamma \le \min\{n-1, \Delta^2+\Delta-1\}$, where $\Delta$ is the maximum degree of a vertex in $G$. 
{\Revise Practically, $\gamma \ll n$ due to the sparse nature of real-world graphs, as verified in the experiments.}
We summarize the time complexity improvements as follows.

\begin{theorem}
    Given a bipartite graph $G$ with $n$ vertices and $m$ edges. The time complexity of \texttt{IPS} with \texttt{IE}-based framework would be $O(n\gamma^2\cdot \alpha^\gamma + \gamma \cdot \beta)$, where $\gamma$ can be computed by Eq.~(\ref{eq:gamma}). 
\end{theorem}
% For \texttt{BPS}, \texttt{BPS+} and \texttt{IPS}, the time complexities would reduce from $O(m\cdot 1.4142^n)$ and  $O(m\cdot 1.3953^n + n \cdot \beta)$ to $O(n\gamma^2 \cdot 1.4142^\gamma)$ and $O(n\gamma^2\cdot 1.3953^\gamma + \gamma \cdot \beta)$, respectively.

\begin{table*}[th]
    \centering
    %\vspace{-2mm}
    \caption{Dataset Statistics. All datasets are reported with $|L|\le |R|$ and are sorted based on $|E|$. }
    %\vspace{-2mm}
    \label{tab:data}
    \begin{tabular}{r|cC{1.5cm}C{1.5cm}C{1.5cm}|ccC{1.2cm}C{1.2cm}}
    \hline
        \textbf{Graph (Name)} & Category & $|L|$ & $|R|$ & $|E|$ & Density & $\beta$ & $\Delta$ & $\gamma$ \\
    \hline
        % \textsf{crime (CR)} & Interaction & 551 & 829 & 1,476 & 2.13 & 620 & 25 & 70 \\
        \textsf{forum (FR)} & Interaction & 522 & 899 & 7,089 & 9.97 &16,261 & 126 & 527\\    
        \textsf{ukwiki (WK)} & Authorship & 779  & 16,699 & 32,615 & 3.73 & 23,772  & 9,254 & 9,734 \\
        \textsf{escorts (ES)} & Rating & 6,625 & 10,107 & 39,044 & 4.66 & 49,538 & 305 & 1,263 \\
    \hline
        \textsf{{youtube (YT)}} & Affiliation & 30,087 & 94,238 & 293,360 & 4.71 & 1,826,587  & 7,591 & 14,947 \\
        \textbf{\textsf{github (GH)}} & Authorship & 56,519 & 120,867 & 440,237 & 4.96 & 55,346,398 & 3,675 & 33,324\\
        \textbf{\textsf{{bkcrossing (BC)}}} & Interaction & 105,278 & 340,523 & 1,149,739 & 5.15 & 54,458,953 & 13,601 & 67,516 \\
        \textsf{stackoverflow (SO)} & Rating & 96,678 & 545,195 & 1,301,942 & 4.05 & 3,320,824 & 6,119 & 36,412 \\
        \textsf{actors (AC)} & Affiliation & 127,823 & 383,640 & 1,470,404 & 5.74 & 1,075,444 & 646 & 4,200\\
        \textbf{\textsf{twitter (TW)}} & Interaction & 175,214 & 530,418 & 1,890,661 & 5.35 & 5,536,526 & 19,805 & 143,233\\
        \textsf{bibsonomy (BS)} & Feature & 204,673 & 767,447 & 2,499,057 & 5.14 & 2,507,269 & 182,673 & 191,911 \\
         \textsf{imdb (IM)} & Affiliation & 303,617  & 896,302 & 3,782,463 & 6.30 & 5,160,061 & 1,590 & 16,133\\
         \textsf{{amazon (AZ)}} & Rating & 1,230,915 & 2,146,057 & 5,743,258 & 3.40 & 7,551,358 & 12,180 & 115,066 \\
        \textbf{\textsf{dblp (DB)}} & Authorship & 1,953,085  & 5,624,219 & 12,282,059 & 3.24 & 4,899,032 & 1,386 & 3,050 \\  
\hline
    \textsf{tvtropes (TV)} & Feature  & 64,415 & 87,678 & 3,232,134  & 42.50 & 19,636,996,096 & 12,400 & 49,893 \\ 
    \textsf{livejournal (LJ)} &  Affiliation & 3,201,203  & 7,489,073  & 112,307,385  & 21.01 & $>10^7$B & 1,053,676 & 3,695,926 \\ 
    \hline
    \end{tabular}
    % \vspace{-2mm}
\end{table*}

\smallskip\noindent
\textbf{Remark.} \underline{(1)} Compared with the algorithm without inclusion-exclusion based technique, when the graph satisfies $n\ge \gamma + 6.01 \cdot \ln \gamma$ (this condition is satisfied by all existing real-world bipartite graphs due to their sparse nature), the time complexities of the algorithms with \texttt{IE} is strictly bounded by those of the algorithms without \texttt{IE}. To see this, 
% we target the case for \texttt{IPS} and omit the case for \texttt{BPS+} as it can be analyzed similarly. Let $\alpha = 1.3953$, 
\begin{equation}
\begin{aligned}
 n \ge \gamma + 6.01 \cdot\ln \gamma \quad & \Rightarrow \quad  n \ln \alpha \ge \gamma \ln \alpha + 6.01\cdot\ln \gamma \ln \alpha \\ 
&\Rightarrow \quad  n \ln \alpha > \gamma \ln \alpha + 2 \ln \gamma \\ 
&\Rightarrow \quad  n \ln \alpha + \ln \frac{m}{n} > \gamma \ln \alpha + 2 \ln \gamma \\ 
% \Rightarrow \quad & \frac{m}{n} \cdot \alpha^n > \gamma^2\cdot \alpha^\gamma \\
&\Rightarrow \quad  m \cdot \alpha^n > n\cdot \gamma^2\cdot \alpha^\gamma \\
\end{aligned}
\end{equation}
The first inequality is an equivalent transformation because a positive value, 
$\ln \alpha$, is multiplied by both sides of the inequality. The second inequality follows from the fact that $6.01\cdot\ln\alpha > 2$. The third inequality holds since a positive value, $\ln(m/n)$, is added to the left side of the inequality. The last inequality is also an equivalent transformation by taking the exponents on both sides of the inequality.
\underline{(2)} As indicated in several subgraph enumeration problems \cite{dai2023hereditary, wang2024efficient, chen2022efficient}, the theoretical and practical performance of the branch-and-bound algorithm are usually sensitive to the vertex ordering. However, it depends in our case. 
Consider the theoretical performance. The vertex ordering would not affect our theoretical outcomes since the exponents mainly depend on $|C|+|X|$ but not solely on $|C|$ in our analysis. Therefore, those vertex ordering algorithms, which target minimizing the largest size of the candidate set after branching on the initial branch, e.g., degeneracy ordering \cite{matula1983smallest} or unilateral ordering \cite{chen2022efficient}, will not improve the theoretical performance of our algorithm. 
Consider the practical performance. As verified in the experiments, we observe that running time varies as different orderings are applied. Therefore, we consider two variants of our algorithms, one using the degeneracy ordering and the other using unilateral ordering in our experiments. The details of the experimental setting will be shown in Section~\ref{sec:exp}.

% Given a graph instance $G_i$ for $1\le i\le |L|$, we construct the branch $B_i$ as follows. 

\section{Experiments}
\label{sec:exp}

\begin{table*}[th]
    % \vspace{-2mm}
    \centering
    \caption{Comparison among variants of our algorithm  (unit: seconds by default). The bold and underlined values indicate the best and the second best, respectively (similar notations are used for other tables). }
    %\vspace{-2mm}
    \label{tab:variant}
    \begin{tabular}{r|C{1.2cm}C{1.2cm}C{1.2cm}|C{1.2cm}C{1.2cm}C{1.2cm}|C{1.2cm}C{1.2cm}C{1.2cm}}
    \hline
    % \multirow{2}{*}{} & \multicolumn{6}{c|}{$k\ge\beta$ (i.e., Maximal Biclique Enumeration)} \\
    % \cline{2-14}
        & \texttt{BCEA} & \texttt{BPS+} & \texttt{IPS} & \texttt{BCEAD} & \texttt{BPSD+} & \texttt{IPSD} & \texttt{BCEAH} & \texttt{BPSH+} & \texttt{IPSH} \\
    \hline
        % \textsf{CR} & 1 & 2 & 3 &  1 & 2 & 3 & 1 & 2 & 3\\
        \textsf{FR} & 29ms & 28ms & 23ms &  21ms & 19ms & \underline{18ms} & 20ms & 18ms & \textbf{16ms}\\    
        \textsf{WK} & 248ms & 192ms & 153ms &  121ms & 73ms & \underline{43ms} & 127ms & 53ms & \textbf{31ms} \\
        \textsf{ES} & 125ms & 118ms & 108ms &  94ms & 77ms & \underline{59ms} & 88ms & 76ms & \textbf{53ms}\\
    \hline
        \textsf{YT} & 17.75 & 6.61 & 5.77 & 7.65 & 2.59 & \textbf{2.29} & 5.92 & 2.59 & \underline{2.31} \\
        \textsf{GH} & 89.63 & 85.79 & 84.40 & 74.00 & 63.44 & \textbf{52.34} & 72.70 & 63.95 & \underline{52.95}\\  
        \textsf{BC} & 324.19 & 184.52 & 166.71 & 308.08 & 172.06 & \underline{129.52} & 303.33 & 161.61 & \textbf{115.09} \\ 
        \textsf{SO} & 62.89 & 28.00 & 23.96 & 146.50 & 18.16 & \underline{13.81} & 97.48 & 15.93 & \textbf{12.59} \\
        \textsf{AC} & 8.71 & 8.46 & 8.01 & 4.82 & 3.50 & \underline{2.72} & 4.94 & 3.52 & \textbf{2.55}\\
        \textsf{TW} & 500.68 & 338.12 & 259.29 & 243.28 & 27.62 & \underline{23.77} & 109.92 & 26.80 & \textbf{18.47}\\
        \textsf{BS} & 85.79 & 80.88 & 43.62 &  17.60 & 16.12 & 15.09 & 17.89 & \underline{14.61} & \textbf{13.13}\\
        \textsf{IM} & 67.37 & 46.24 & 43.86 & 57.31 & 18.32 & \textbf{17.23} & 52.91 & 19.02 & \underline{18.20}\\
        \textsf{AZ} & 594.11 & 190.69 & 136.75 & 471.45 & 156.72 & \underline{54.23} & 236.60 & 132.50 & \textbf{35.08} \\
        \textsf{DB} & 16.00 & 15.72 & 12.38 &  13.62 & 11.37 & \textbf{7.92} & 14.92  & 11.85 & \underline{8.03}\\
    \hline
        % \textsf{TV} & - (13B) & - () & - () & 86134.47 & 34962.46 & \textbf{33259.40} & 85104.24 & 36494.63 & \underline{34923.54}\\
        % \textsf{LJ} & - (12B) & - () & - () & - (16B) & - (33B) & \underline{- (45B)} & - (18B) & - (35B) & \textbf{- (75B)}\\ 
        \textsf{TV} & - (13B) & 20.7h & 18.2h & 23.9h & 9.7h & \textbf{9.2h} & 23.6h & 10.1h & \underline{9.7h} \\
        \textsf{LJ} & - (12B) & - (17B) & - (29B) & - (16B) & - (33B) & \underline{- (65B)} & - (18B) & - (35B) & \textbf{- (75B)}\\ 
    \hline
    \end{tabular}
       % \vspace{-2mm}
\end{table*}

\subsection{Experimental setup}
\label{subsec:setup}

\smallskip\noindent
\textbf{Datasets.} We use 15 real datasets in our experiments, which can be obtained from \url{http://konect.cc/}. For each graph, we ignore the weights and multiple edges~\footnote{The multiple edges represent those that are duplicated themselves multiple times in a graph.} at the very beginning. 
Given a graph $G$, in addition to the number of vertices and edges inside, we also collect the statistics and report the edge density, which is defined as $\rho = 2\times |E|/(|L|+|R|)$, the maximum degree $\Delta$, the number of all maximal bicliques $\beta$, and the value of $\gamma$ introduced in Section~\ref{sec:boost}, which are summarized in Table~\ref{tab:data}. We use four representative datasets as default ones, shown in bold in Table~\ref{tab:data}, i.e., \textsf{GH}, \textsf{BC}, \textsf{TW} and \textsf{DB}, since they cover different graph scales. 

\smallskip\noindent
\textbf{Baselines.} \underline{First}, we introduce three variants of our algorithms. 
Specifically, we first consider the simple version \texttt{IPS} (i.e., the branch-and-bound algorithms without \texttt{IE}). In addition, for \texttt{IPS} with \texttt{IE}, recall that in Section~\ref{sec:boost}, we consider using degeneracy ordering and unilateral ordering to branch the initial branch. Therefore, we denote \texttt{IPS} with degeneracy ordering and unilateral ordering by \texttt{IPSD} and \texttt{IPSH}, respectively. 
% Since this study addresses the top-$k$ maximum biclique search problem, we compare our algorithm with various baselines for different values of $k$. The details are elaborated as follows. 
% \underline{First}, when $k\ge \beta$, the problem reduces to maximal biclique enumeration problem. Therefore, 
\underline{Second}, we consider five existing algorithms, namely \texttt{BCEAD} \cite{dai2023hereditary}, \texttt{BCEAH} \cite{dai2023hereditary}, \texttt{BCDelay} \cite{dai2023hereditary}, \texttt{MBET} \cite{chen2023maximal} and \texttt{MBETM} \cite{chen2023maximal}. The first two algorithms are the branch-and-bound algorithms adopting the basic pivot selection strategy, with variations in vertex ordering. The latter three baselines are the branch-and-bound algorithms that achieve polynomial delay. They differ primarily in their implementation details and the space organization.

\smallskip\noindent
\textbf{Settings.} When enumerating all maximal bicliques (i.e., $\tau_L=\tau_R=1$), we compare our algorithms \texttt{IPSD} and \texttt{IPSH} with all baseline methods. For size-constrained maximal bicliques (i.e., $\tau_L>1$ and/or $\tau_R>1$), we exclude \texttt{MBET} and \texttt{MBETM} from consideration because the prefix tree structure proposed in \cite{chen2023maximal} is incompatible with the pruning rules introduced in Section~\ref{subsec:summary}, i.e., size-constrained maximal bicliques can only be obtained after enumerating all maximal bicliques and filtering out the ineligible ones. All algorithms are implemented in C++, and experiments are conducted on a Linux machine with a 2.10GHz Intel CPU and 128GB memory. We set the time limit to 24 hours (i.e., 86,400 seconds). For any algorithm that exceeds this time limit, (1) the running time is recorded as ``-'' or ``INF'' and (2) the number of maximal bicliques that have already been enumerated within the 24-hour period is reported. For each algorithm, we report the average result over 5 runs. Implementation codes and datasets can be found via this link \url{https://github.com/wangkaixin219/IPS}.

\subsection{Experimental results}
\label{subsec:result}

\begin{table}[t]
\small
    \centering
    \caption{Comparison among all baselines when enumerating all maximal bicliques (unit: seconds by default). ``$\times$'' indicates that the algorithm is terminated due to the ``Out Of Memory''. }
    %\vspace{-2mm}
    \label{tab:baseline}
    \begin{tabular}{r|ccccccc}
    \hline
        & \texttt{IPSH} & \texttt{IPSD} & \texttt{BCEAH} & \texttt{BCEAD} & \texttt{BCDelay} & \texttt{MBET} & \texttt{MBETM} \\
    \hline
        % \textsf{CR} &  &  &  &   &  &  & \\
        \textsf{FR} & \textbf{16ms} & \underline{18ms} & 20ms & 21ms & 26ms & 19ms & 19ms \\    
        \textsf{WK} & \textbf{31ms} & \underline{43} & 127ms & 121ms & 170ms & 62ms & 66ms \\
        \textsf{ES} & \textbf{53ms} & \underline{59ms} & 88ms & 94ms & 98ms & 72ms & 79ms \\
    \hline
        \textsf{YT} & \underline{2.31} & \textbf{2.29} & 5.92 & 7.65 & 8.47 & 5.06 & 5.24 \\
        \textsf{GH} & \underline{52.95} & \textbf{52.34} & 72.70 & 74.00 & 81.29 & 202.49 & 255.01 \\
        \textsf{BC} & \textbf{115.09} & \underline{129.52} & 303.33 & 308.08 & 335.72 & 219.32 & 222.84 \\
        \textsf{SO} & \textbf{12.59} & 13.81 & 97.48 & 146.50 & 239.18 & \underline{12.74} & 16.90 \\
        \textsf{AC} & \textbf{2.55} & \underline{2.72} & 4.94 & 4.82 & 5.94 & 3.33 & 3.30  \\
        \textsf{TW} & \textbf{18.47} & \underline{23.77} & 109.92 & 243.28 & 192.37 & 48.86 & 57.11 \\
        \textsf{BS} & \underline{13.13} & 15.09 & 17.89 & 17.60 & 22.08 & \textbf{12.21} & 14.43 \\
        \textsf{IM} & \underline{18.20} & \textbf{17.23} & 52.91 & 57.31 & 60.75 & 22.97 & 25.20 \\
        \textsf{AZ} & \textbf{35.08} & 54.23 & 236.60 & 471.45 & 892.14 & \underline{48.34} & 59.04 \\
        \textsf{DB} & \underline{8.03} & \textbf{7.92} & 14.92 & 13.62 & 12.63 & 15.29 & 18.32 \\
    \hline
        \textsf{TV} & \underline{9.7h} & \textbf{9.2h} & 23.6h & 23.9h & - (17B) & $\times$ & - (6B) \\
        \textsf{LJ} & \textbf{- (75B)} & \underline{- (65B)} & - (18B) & - (16B) & - (15B) & $\times$ & - (2B) \\
        % \textsf{TV} & \underline{34923.54} & \textbf{33259.40} & 85104.24 & 86134.47 & - (13.3B) & $\times$ & - (5.7B) \\
        % \textsf{LJ} & - (75.3B) & - (45.3B) & - () & - () & - () & $\times$ & - (2.1B) \\
    \hline
    \end{tabular}
    %\vspace{-4mm}
\end{table}

\begin{figure}[t]
%\vspace{-2mm}
\begin{center}
    \includegraphics[width=0.8\linewidth]{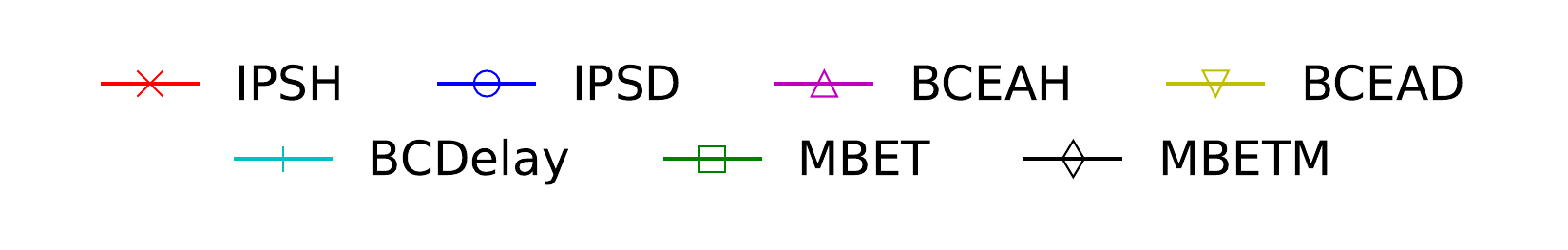}
\end{center}
\vspace{-4mm}
\subfigure[\textsf{GH}]{
    \includegraphics[width=0.23\textwidth]{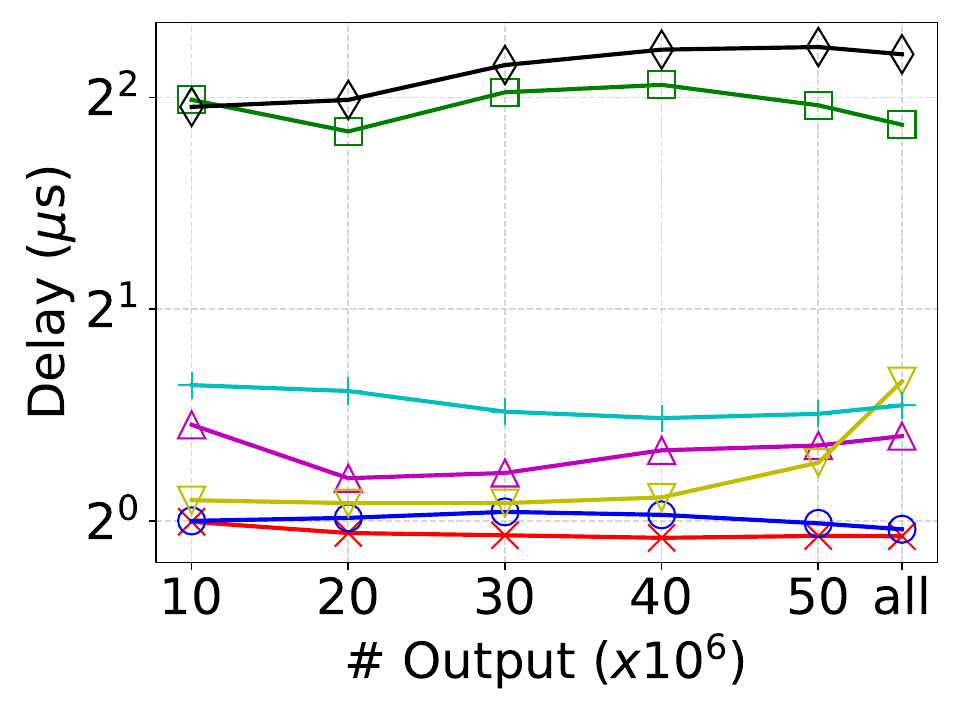}}
\subfigure[\textsf{BC}]{
    \includegraphics[width=0.23\textwidth]{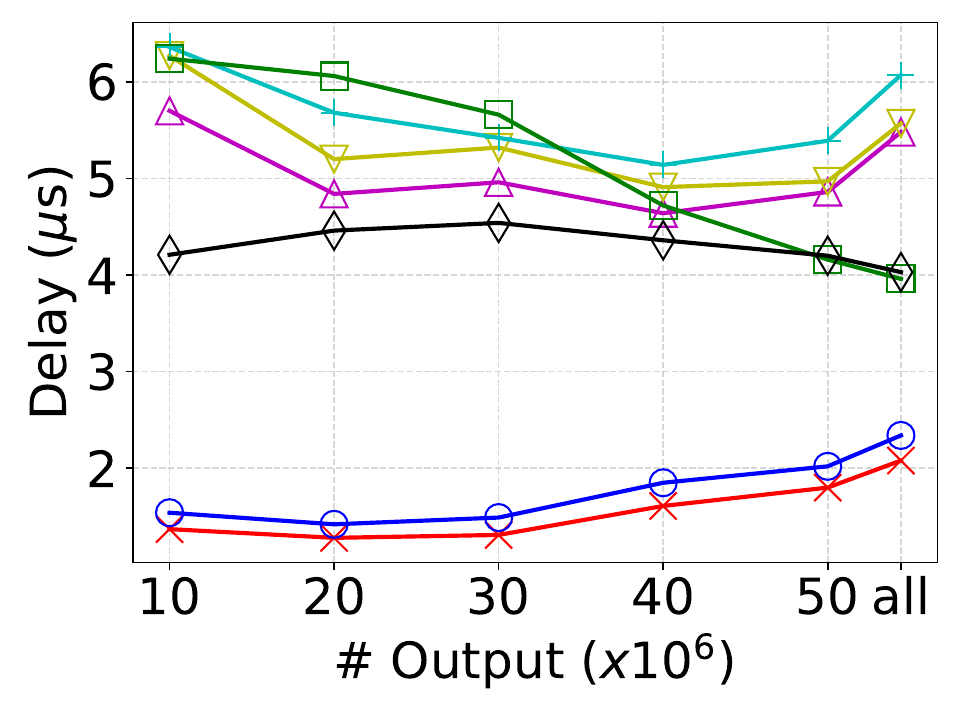}}
\subfigure[\textsf{TW}]{
    \includegraphics[width=0.23\textwidth]{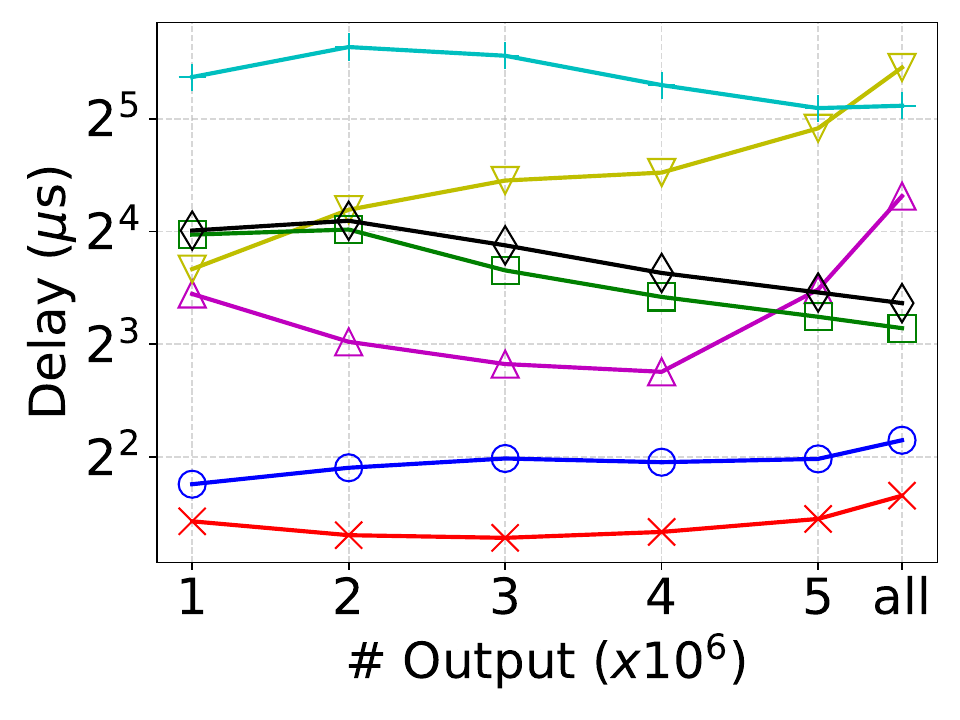}}
\subfigure[\textsf{DB}]{
    \includegraphics[width=0.23\textwidth]{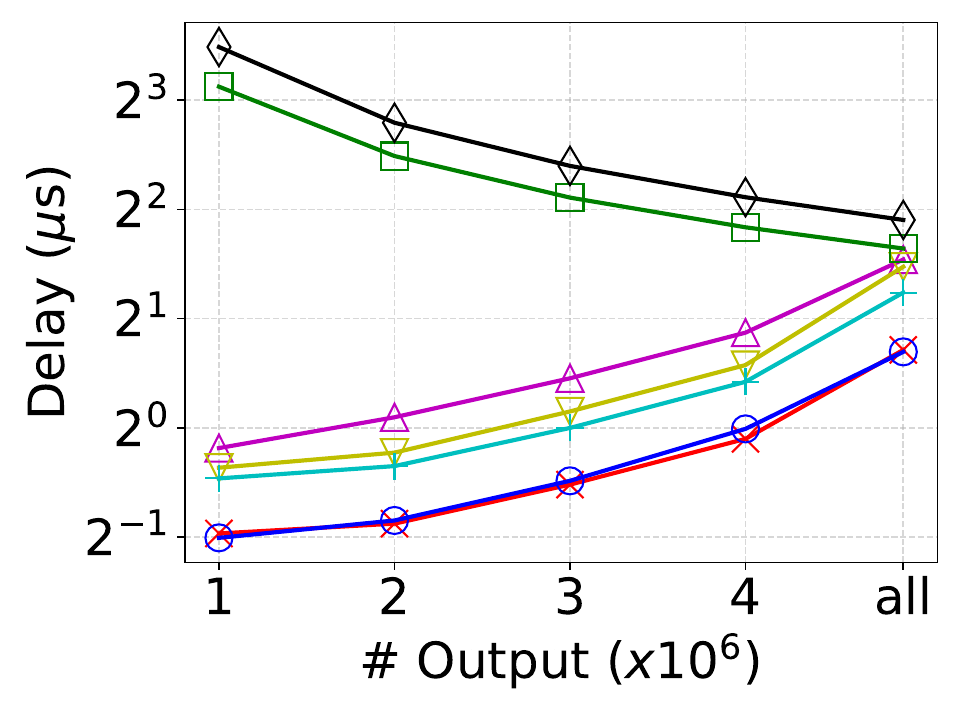}}
%\vspace{-2mm}
\caption{Results of delay testing.}
%\vspace{-4mm}
\label{fig:delay}
\end{figure}

\smallskip\noindent
\textbf{(1) Comparison among variants (real datasets).}
In addition to the baselines introduced in Section~\ref{subsec:setup}, we further include several BB algorithms to verify the effectiveness of each proposed technique: (1) \texttt{BCEA} \cite{dai2023hereditary} and (2) \texttt{BPS+}, \texttt{BPSD+} and \texttt{BPSH+} correspond to the algorithms \texttt{BCEA}, \texttt{BCEAD} and \texttt{BCEAH} with the conventional pivoting strategy and our proposed new stopping criterion. 
Table~\ref{tab:variant} shows the results of the comparison among the different variants of our algorithms.
% in which \texttt{BCEA} \cite{dai2023hereditary} is equivalent to \texttt{BPS} without \texttt{IE}. 
We have the following observations. 
\underline{First}, compare \texttt{BPS+} with \texttt{BCEA}. \texttt{BPS+} outperforms \texttt{BCEA} on all datasets, which shows the effectiveness of the new stopping criterion technique. The results are consistent if we compare \texttt{BPSD+} and \texttt{BCEAD} (resp. \texttt{BPSH+} and \texttt{BCEAH}). 
\underline{Second}, compare \texttt{IPS} with \texttt{BPS+}. \texttt{IPS} runs faster than \texttt{BPS+} on all datasets, which shows the effectiveness of the new pivoting strategy. We also collect the statistics of the produced branches. For example, on \texttt{YT} dataset, \texttt{BCEA} produces 8.28M branches. With the new stopping criterion technique, \texttt{BPS+} reduces the number of branches to 5.17M branches.
% , among which 41.5\% branches can be early-terminated. 
\texttt{IPS} further reduces the number of branches to 4.66M.
% , in which 48.9\% branches can be early-terminated. 
This result further confirms that the new pivoting strategy is more compatible with the proposed stopping criterion compared to the conventional pivoting strategy. Together, these two techniques significantly contribute to the observed efficiency improvements.
The results are also consistent if we compare \texttt{IPSD} and \texttt{BPSD+} (resp. \texttt{IPSH} and \texttt{BPSH+}). Notably, \texttt{IPSD} can achieve up to an order of magnitude speedup over \texttt{BCEAD} on \texttt{TW} dataset. 
\underline{Third}, compare \texttt{IPS} with \texttt{IPSD} (or \texttt{IPSH}). The latter ones alway runs faster than the former one, which verifies the effectiveness of the \texttt{IE} framework. Besides, \texttt{IPSH} and \texttt{IPSD} have comparable results on the majority of the datasets since they have the same worst-case time complexities, but differ with each other on some datasets such as \texttt{BC} and \texttt{TW}. 
The possible reasons are (1) generating degeneracy ordering and the unilateral ordering would incur different time costs, i.e., $O(m)$ for the former and $O(\Delta\cdot m)$ for the latter, respectively, and (2) the empirical time costs for enumeration mainly depends on the size of $|C|$ since the depth of the recursive tree is at most $|C|$. As indicated in \cite{chen2022efficient}, unilateral ordering would bring more benefits than degeneracy ordering, i.e., the maximum $|C_i|$ of a sub-branch $B_i$ produced from the initial branch is smaller. 
% Therefore, it is necessary to make a trade-off between two orderings. Our recommendation is that when $|E|$ is less than 10M, it is better to use unilateral ordering; otherwise, using degeneracy ordering is better. 

\begin{figure}[t]
%\vspace{-2mm}
\subfigure[\textsf{GH}]{
    \includegraphics[width=0.23\textwidth]{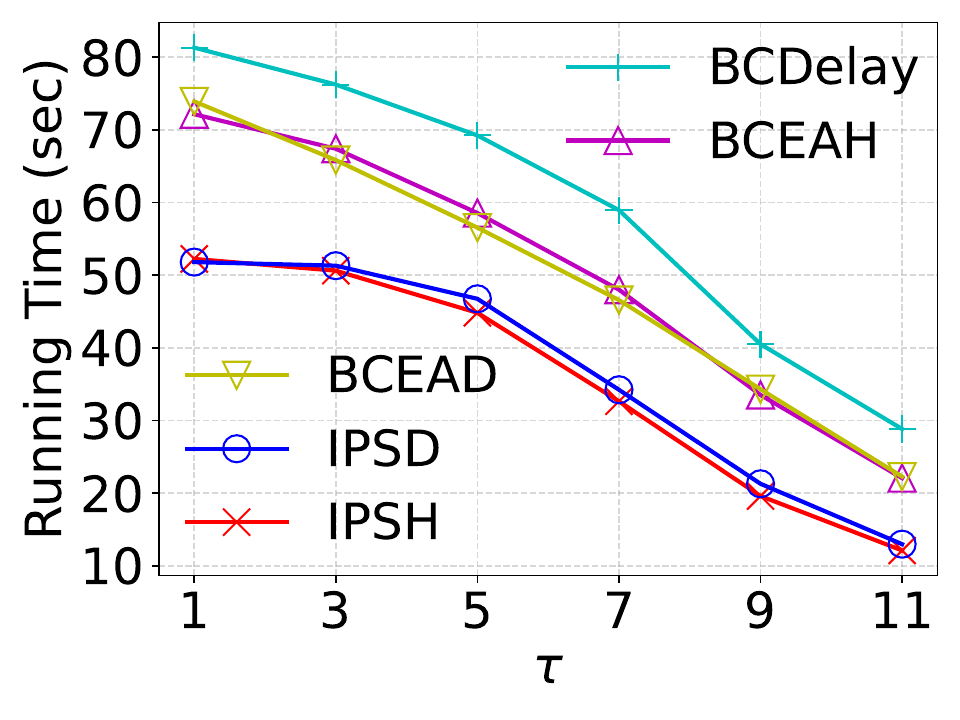}}
\subfigure[\textsf{BC}]{
    \includegraphics[width=0.23\textwidth]{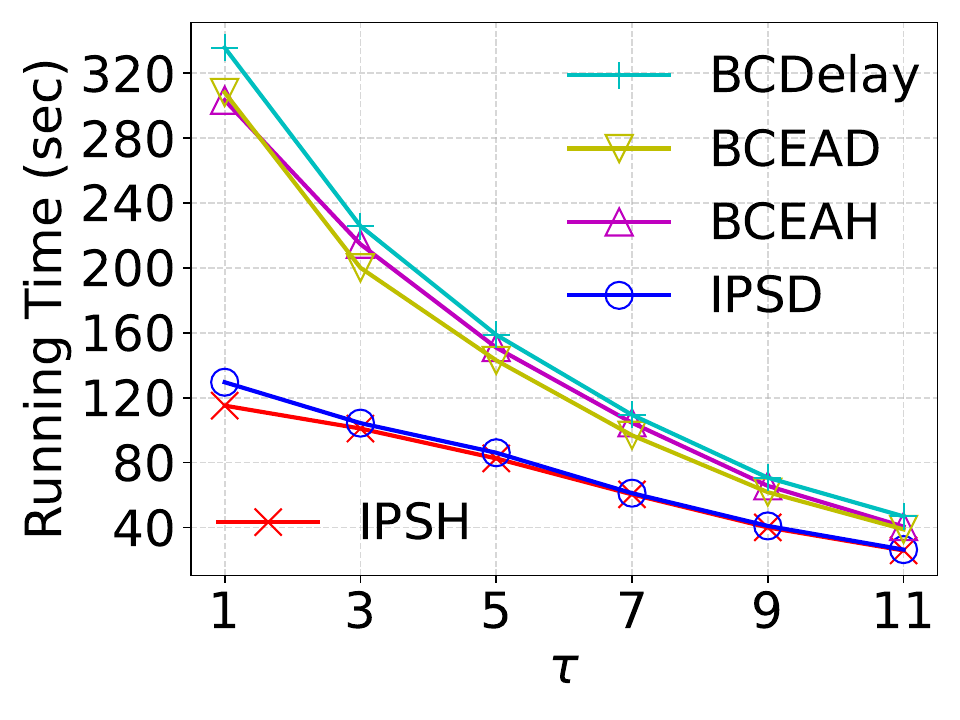}}
\subfigure[\textsf{TW}]{
    \includegraphics[width=0.23\textwidth]{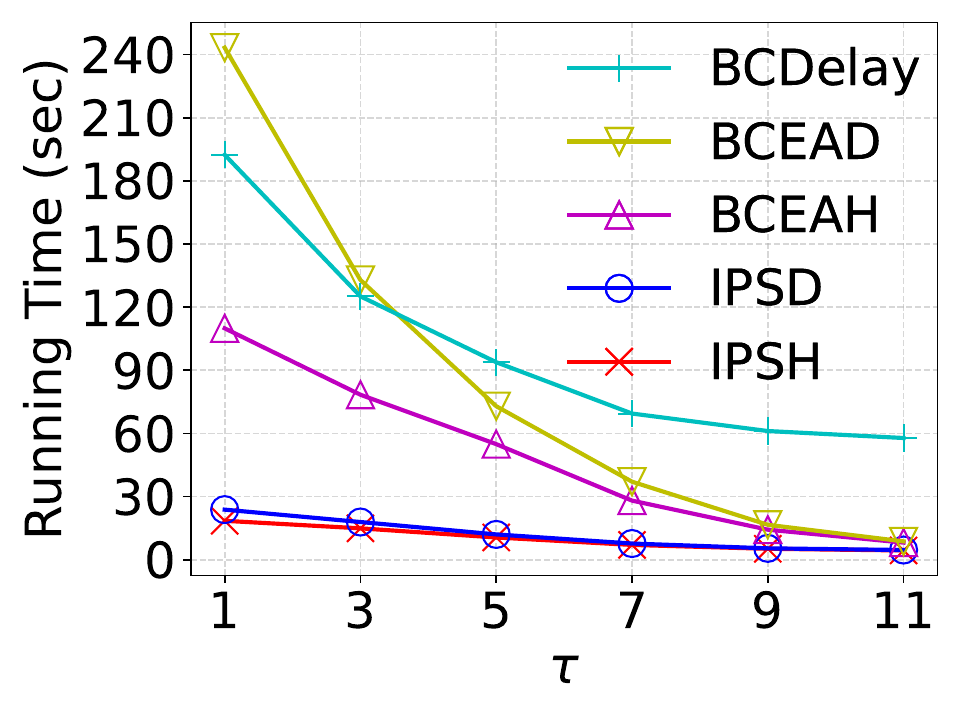}}
\subfigure[\textsf{DB}]{
    \includegraphics[width=0.23\textwidth]{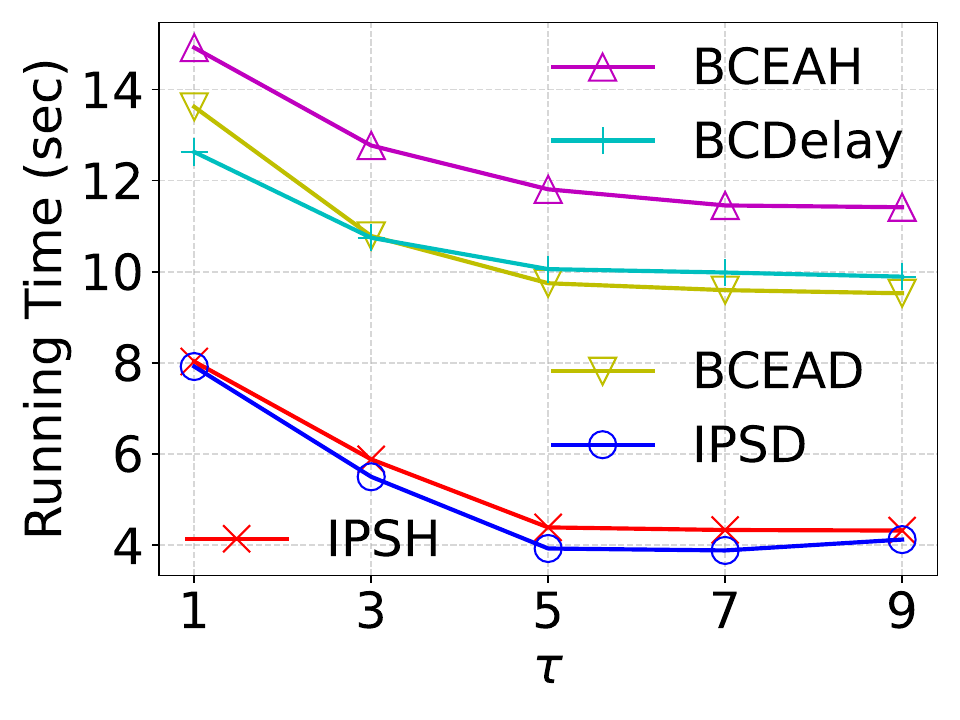}}
%\vspace{-2mm}
\caption{Comparison among baselines when enumerating size-constrained maximal bicliques (varying $\tau=\tau_L=\tau_R$).}
\vspace{-3mm}
\label{fig:tau-maximal}
\end{figure}

\smallskip\noindent
\textbf{(2) Comparison among all baselines when enumerating all maximal bicliques (real datasets).} 
Table~\ref{tab:baseline} shows the results of the comparison among all baselines when enumerating all maximal bicliques. We observe that our proposed algorithm \texttt{IPSH} and \texttt{IPSD} can run faster than other algorithms on most datasets. Specifically, for all datasets except \textsf{TV} and \textsf{LJ}, \texttt{IPSH} can achieve an average 3.2x (up to 7.7x) speedup and 1.7x (up to 3.8x) speedup over \texttt{BCEAH} and \texttt{MBET}, respectively. This is consistent with our theoretical analysis that the worst-case running time of \texttt{IPSH} and \texttt{IPSD} is smaller than that of other algorithms. Moreover, on the hard datasets (i.e., \textsf{TV} and \textsf{LJ}), our algorithm can largely outperform the existing algorithms by enumerating much more maximal bicliques than the other algorithms when all algorithms cannot finish the task within the time limit. On \textsf{BS}, our algorithms are not so efficient as \texttt{MBET}. The possible reason is that \texttt{MBET} uses the bitmap to compress the neighborhood information and compute the intersection by performing a bitwise AND operation. We do not utilized this data representation in our implementation, despite the potential to further enhance the acceleration of our algorithm.

\smallskip\noindent
\textbf{(3) Delay test for enumerating all maximal bicliques (real datasets).}
Following \cite{dai2023hereditary}, the delay is evaluated by dividing the running time for reporting the first $k$ maximal bicliques by $k$. Figure~\ref{fig:delay} shows the results of the delay testing by varying $k$. \texttt{IPSH} and \texttt{IPSD} have the smallest delays across different datasets and their delays do not change a lot when we vary the number of the outputs. The possible reason is that \texttt{IPSH} and \texttt{IPSD} can output the maximal bicliques in a batch fashion when the condition for the new stopping criterion is satisfied, i.e., the delay for the maximal bicliques in the same batch is $O(\gamma)$, which is the optimal since we need at least $O(\gamma)$ to output a maximal biclique. This compensates for the fact that the delay between two consecutive batches could be exponential. Moreover, \texttt{IPSH} and \texttt{IPSD} outperform other algorithms when outputting a fixed number of results, which indicates that our algorithms are also suitable for the applications that only require to output a fixed number of results.

\smallskip\noindent
\textbf{(4) Comparison among baselines when enumerating size-constrained maximal bicliques (real datasets).}
Figure~\ref{fig:tau-maximal} shows the result of varying $\tau=\tau_L=\tau_R$ for maximal biclique enumeration, respectively. As expected, our proposed \texttt{IPSH} and \texttt{IPSD} outperform other algorithms under most settings. Meanwhile, as the value of $\tau$ increases, the running time of all algorithms decreases. The reason is that the pruning rules become more powerful as $\tau$ increases and they can dramatically prune many unpromising branches so as to reduce the search space.

\begin{figure}[t]
\subfigure[Varying \% of vertices (\textsf{TV})]{
    \includegraphics[width=0.35\textwidth]{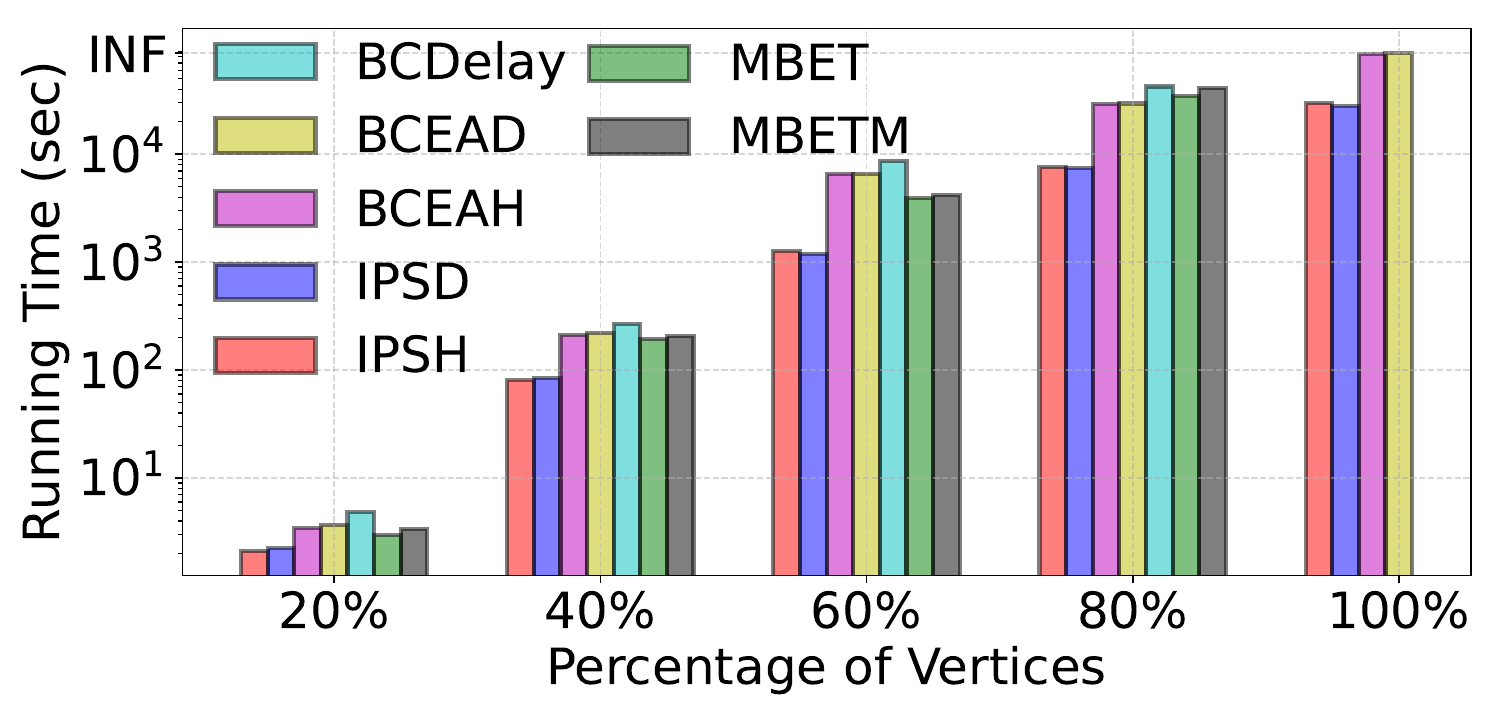}}
\subfigure[Varying \% of edges (\textsf{TV})]{
    \includegraphics[width=0.35\textwidth]{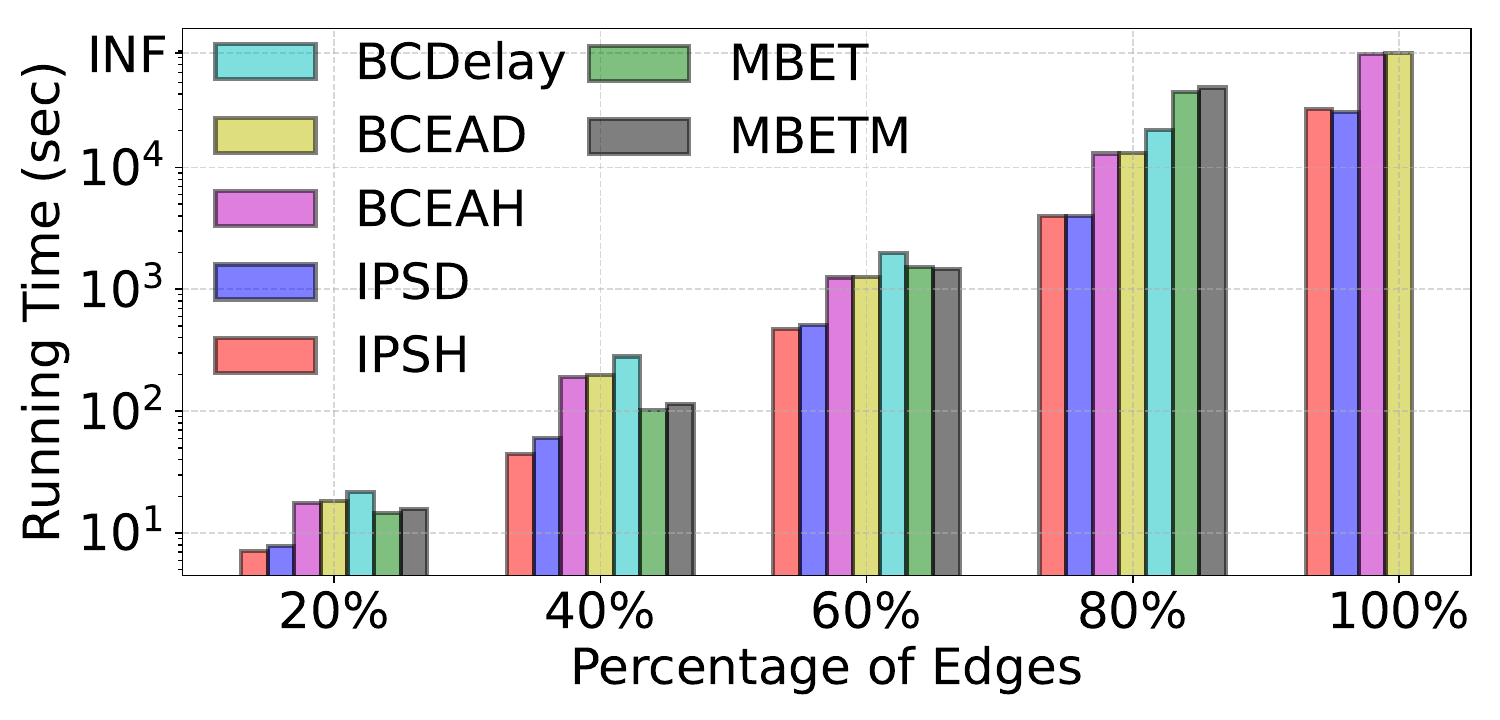}}
% \subfigure[Varying \% of vertices (\textsf{LJ})]{
%     \includegraphics[width=0.23\textwidth]{backup/figs/empty.png}}
% \subfigure[Varying \% of edges (\textsf{LJ})]{
%     \includegraphics[width=0.23\textwidth]{backup/figs/empty.png}}
%\vspace{-2mm}
\caption{Results of scalability testing on \textsf{TV} dataset.}
\vspace{-3mm}
\label{fig:scale}
\end{figure}

\begin{table}[t]
\small
    \centering
    \caption{Space costs (unit: MB by default). }
    %\vspace{-2mm}
    \label{tab:space}
    \begin{tabular}{r|ccccccc}
    \hline
       & \texttt{IPSH} & \texttt{IPSD} & \texttt{BCEAH} & \texttt{BCEAD} & \texttt{BCDelay} & \texttt{MBET} & \texttt{MBETM} \\
    \hline
    \textsf{FR} & 6.53 & 6.53 & \underline{6.40} & \textbf{6.40} & 6.40 & 7.82 & 7.56 \\ 
        \textsf{WK} & 10.58 & 10.59 & \underline{9.45} & \textbf{9.21} & 9.58 & 22.32 & 20.97 \\
        \textsf{ES} & 9.91 & 9.89 & 8.87 & \textbf{8.85} & \underline{8.86} & 16.59 & 15.68 \\
    \hline
        \textsf{YT} & 35.79 & 36.64 & \underline{28.33} & \textbf{28.25} & 28.40 & 180.18 & 114.19 \\
        \textsf{GH} & 47.57 & 48.16 & \underline{38.10} & \textbf{37.90} & 38.83 & 2.04G & 311.48 \\
        \textsf{BC} & 106.86 & 105.10 & \underline{90.27} & 90.60 & \textbf{89.97} & 4.29G & 869.92 \\
        \textsf{SO} & 152.56 & 160.99 & \textbf{113.57} & \underline{113.72} & 113.74 & 425.65 & 380.91 \\
        \textsf{AC} & 131.62 & 136.31 & \underline{99.94} & \textbf{99.86} & 100.01 & 371.86 & 328.82 \\
        \textsf{TW} & 179.06 & 203.48 & \underline{140.37} & \textbf{139.21} & 141.22 & 639.70 & 442.07 \\
        \textsf{BS} & 285.43 & 269.72 & \underline{197.57} & 197.73 & \textbf{197.14} & 680.96 & 625.87 \\
        \textsf{IM} & 317.99 & 329.87 & \underline{239.94} & \textbf{239.58} & 240.13 & 839.36 & 803.77 \\
        \textsf{AZ} & 691.27 & 679.07 & \textbf{512.69} & \underline{513.07} & 514.91 & 1.90G & 1.65G \\
        \textsf{DB} & 1.77G & 1.72G & \underline{1.12G} & \textbf{1.12G} & 1.12G & 3.56G & 3.78G \\
        \hline
    \end{tabular}
    \vspace{-3mm}
\end{table}

\smallskip\noindent
\textbf{(5) Scalability test (real datasets).} We conduct the scalability test on one hard dataset \textsf{TV}. Following existing studies \cite{dai2023hereditary, chen2023maximal}, we randomly sample 20\%-80\% vertices and edges from the dataset and build the induced subgraph. We show the results in Figure~\ref{fig:scale}. We have the following observations. First, as the number of the sampled vertices/edges increases, the running time of all algorithms increases, among which \texttt{IPSH} and \texttt{IPSD} run the fastest. Besides, our algorithm can enumerate all these 19 billion maximal bicliques in this hard dataset within INF while other algorithms cannot.

\smallskip\noindent
\textbf{(6) Space costs (real datasets).}
We report the space costs of different algorithms in Table~\ref{tab:space}. We have the following observations. First, our algorithms use slightly more memory than \texttt{BCEAD} and \texttt{BCEAH} to enumerate the maximal bicliques. The reason is that our algorithms need additional space for each branch $B=(S,C,X)$ to record the non-neighborhood information to check whether $v$ should be skipped. Second, our algorithms use much less memory than \texttt{MBET} and \texttt{MBETM}. The reason is that \texttt{MBET} and \texttt{MBETM} store and update the whole prefix tree in an in-memory fashion, which increases dramatically as the algorithm proceeds.

% \clearpage
\section{Related Work}
\label{sec:related}

% There are two categories of problems, namely the maximum biclique search problem and the maximal biclique enumeration problem, which are the most related topics to the problem studied in this paper. In addition, we also review several biclique-related problems.
% , e.g., $(p,q)$-biclique listing problem and top-$k$ diversified maximum biclique search problem. 
% We review each of them as follow. 

% , each of which is specified by parameters $\tau_L$ and $\tau_R$. These subproblems are derived , and in each subproblem, it restricts the search in a very small part of the original bipartite graph.

% is to execute a branch-and-bound algorithm multiple rounds, where the algorithm will find the maximum biclique in a reduced graph in each round. The more the round is, the tighter the bound is. 

% we divide the problem into several subproblems each of which
% is specified using two parameters. 

\smallskip\noindent
\textbf{Maximal Biclique Enumeration.} The maximal biclique enumeration problem has also been studied for decades. Earlier studies \cite{zaki2002charm, uno2004lcm, li2007maximal} enumerate all maximal bicliques in a breath-first search manner, where they first build the powerset of one side of the bipartite graph, e.g., $L$, and then extract those that induce the maximal bicliques. Recent studies \cite{makino2004new, gely2009enumeration, zhang2014finding, abidi2020pivot, chen2022efficient, dai2023hereditary, chen2023maximal} usually enumerates the maximal bicliques in a depth-first search manner as they follow the Bron-Kerbosch branch-and-bound framework \cite{bron1973algorithm}, which is originally designed for maximal clique enumeration problem. 
Specifically, \cite{makino2004new, gely2009enumeration} reduce the problem to the maximal clique enumeration problem by transforming the bipartite graph into a general graph.
\cite{zhang2014finding} is the first study that adapts the Bron-Kerbosch branching strategy to enumerate maximal biclique directly. Motivated by the phenomenon that the pivoting strategy \cite{tomita2006worst} can prune a significant number of branches, \cite{abidi2020pivot, dai2023hereditary} both employ the technique in biclique scenario but differ in which pivot is selected. \cite{chen2022efficient} aims to optimize the order of vertices, introducing the unilateral vertex ordering to accelerate the enumeration. \cite{chen2023maximal} further improves the efficiency and reduces the memory usage by using the prefix-tree to store the lists of elements compactly. However, all these branch-and-bound algorithms differ from \texttt{IPS} in that \texttt{IPS} involves a new pivoting strategy, which would theoretically reduce the number of created branches under the worst case, i.e., $O(\alpha^n)$ branches for \texttt{IPS} v.s. $O(\sqrt{2}^n)$ branches for all existing studies with $\alpha < \sqrt{2}$. This new pivot strategy also makes \texttt{IPS} and its variants have a strictly better time complexity than that of existing studies.

\smallskip\noindent
\textbf{Other Biclique-related Problems.} There are also several other interesting topics that involve the biclique inside. 
\underline{First}, several biclique enumeration tasks impose the size constraints. \cite{yang2021p} aims to find all $(p,q)$-bicliques in a graph, i.e., each returned biclique $H$ satisfies $|L(H)|=p$ and $|R(H)|=q$, where $p$ and $q$ are user-specified parameters. It extends the ideas for $k$-clique listing to this setting with several optimizations such as preallocated arrays and vertex renumbering. \cite{wang2019vertex, wang2023accelerated} target a special $(p,q)$-biclique with $p=q=2$, a.k.a, butterfly, in the bipartite graph and design specialized algorithm for this task. 
\underline{Second}, diversified top-$k$ search is another topic that has been extensively studied, which aims to find top-$k$ results that are not only most relevant to a query but also diversified. In the biclique enumeration setting,
\cite{lyu2022maximum} studies the top-$k$ diversified maximum biclique search, which aims to find $k$ bicliques such that they cover the most number of edges. It extends the idea for maximum biclique search \cite{lyu2020maximum} by iteratively identifying and removing the  maximum biclique in the remaining graph. \cite{abidi2022maximising} extends this problem with a size constraint $t$ for one side, i.e., $|L(H)|=t$, while maximizing coverage on the other side. 
\underline{Third}, the maximum biclique search problem has also gained significant research interest due to its wide range of applications in fields such as bioinformatics and social network analysis. There are three lines of research - (1) \cite{garey1979computers} targets vertex-based maximum biclique search problem, which aims to  find a biclique $H^*$ such that $|V(H^*)|$ is maximized, (2) \cite{lyu2020maximum, lyu2022maximum, shaham2016finding, sozdinler2018finding} target edge-based maximum biclique search problem, which aims to find a biclique $H^*$ such that $|E(H^*)|$ is maximized and (3) \cite{chen2021efficient, mccreesh2014exact, wang2018new, yuan2014new, zhou2017combining, zhou2018towards} target the maximum balanced biclique search problem, which aims to find a biclique $H^*$ such that $|L(H^*)|=|R(H^*)|$ and $|E(H^*)|$ is maximized.  
We note that the techniques used in the above studies are not suitable for our problem setting since (1) we do not impose an exact size constraint for the returned bicliques, (2) the bicliques returned can be overlapped for our problem and (3) we aim to find all maximal bicliques but not only the maximum one.

% \cite{shahinpour2017scale}
\section{Conclusion}
\label{sec:conclusion}

In this paper, we study the maximal biclique enumeration problem. We propose the algorithm \texttt{IPS}, which employs two techniques to improve performance. The first is a new termination condition, designed to find all maximal bicliques in a branch in nearly optimal time. The second is a new pivot selection strategy that aligns well with the early-termination technique. \texttt{IPS} achieves better worst-case time complexity than all existing algorithms for maximal biclique enumeration. Extensive experiments on several datasets demonstrate the efficiency of our algorithm and the effectiveness of the proposed techniques. An interesting research direction is to generalize the idea of the new pivoting strategy for clique-related problems, such as maximal clique enumeration, maximum clique search, and top-$k$ diversified clique search.

\begin{acks}
We would like to thank the anonymous reviewers for providing constructive feedback and valuable suggestions. This work was supported by the National Natural Science Foundation of China (Grant No. 62506018).
This work was also supported in part by the Ministry of Education, Singapore, under its Academic Research Fund (Tier 2 Award MOE-T2EP20224-0011 and Tier 1 Award (RG20/24)). Any opinions, findings and conclusions or recommendations expressed in this material are those of the author(s) and do not reflect the views of the Ministry of Education, Singapore.
\end{acks}

% \clearpage
% \balance
\bibliographystyle{ACM-Reference-Format}
\bibliography{ref}

% \clearpage
% \input{appendix}

% \clearpage

\end{document}